\documentclass[12pt,a4paper,oneside]{report}

\usepackage[utf8]{inputenc}
\usepackage[T1]{fontenc}
\usepackage{geometry}
\usepackage{graphicx}
\usepackage{amsmath, amssymb, amsthm}
\usepackage{hyperref}
\usepackage{natbib}
\usepackage{tikz}
\usepackage{booktabs}
\usepackage{array}
\usepackage{enumitem}
\usepackage{setspace}
\usepackage{caption}
\usepackage{subcaption}
\usepackage{float}
\usepackage{xcolor}
\usepackage{mdframed}
\usepackage{microtype}
\usepackage{fancyhdr}
\usepackage{emptypage}
\usepackage{titlesec}

\hypersetup{
  colorlinks=true,
  linkcolor=blue!60!black,
  citecolor=blue!60!black,
  urlcolor=blue!60!black,
  pdfauthor={Sergio Montenegro},
  pdftitle={Programmable Participatory Governance}
}

\titleformat{\chapter}[display]
  {\normalfont\huge\bfseries}{\chaptertitlename\ \thechapter}{20pt}{\Huge}
\titlespacing*{\chapter}{0pt}{40pt}{20pt}

\fancypagestyle{plain}{%
  \fancyhf{}%
  \fancyfoot[C]{\small\thepage}%
}

\newtheorem{definition}{Definition}[chapter]
\newtheorem{proposition}{Proposition}[chapter]
\newtheorem{theorem}{Theorem}[chapter]

\title{%
  \textbf{Programmable Participatory Governance:}\\[4pt]
  A Formal, Scalable, and Empirically-Grounded Framework\\
  for Democratic Systems
}

\author{%
  Sergio Montenegro\\[4pt]
  \small Independent Researcher\\
  \small\href{https://orcid.org/0009-0001-1869-9580}{ORCID: 0009-0001-1869-9580}
}

\date{April 2026}

\begin{document}

% ── Title Page ───────────────────────────────────────────────────────────────
\begin{titlepage}
\centering
\vspace*{2cm}
{\LARGE\bfseries Programmable Participatory Governance}\\[1em]
{\large\itshape A Formal Framework for Transparent, Accountable, and Citizen-Responsive\\
Democratic Systems: From Deliberative Theory to Decentralised Architecture}\\[3cm]
\rule{0.7\textwidth}{0.4pt}\\[0.5em]
{\normalsize
\begin{tabular}{c}
\textbf{Sergio Montenegro} \\
\small Independent Researcher \\
\small\href{https://orcid.org/0009-0001-1869-9580}{\texttt{orcid.org/0009-0001-1869-9580}}
\end{tabular}
}\\[0.5em]
\rule{0.7\textwidth}{0.4pt}\\[2.5cm]
% {\normalsize\textit{Submitted in partial fulfilment of the requirements for\\
% a research thesis in governance, digital democracy, and distributed systems}}\\[0.6em]
{\normalsize April 2026}\\[3cm]
\vfill
\begin{footnotesize}
This work is released under the Creative Commons Attribution 4.0 International Licence (CC BY 4.0).\\[0.3em]
Source code, simulation data, and all supplementary materials:\\[0.2em]
\href{https://github.com/seht/participatory-governance}{github.com/seht/participatory-governance}
\end{footnotesize}
\end{titlepage}

% ── Declaration ───────────────────────────────────────────────────────────────
\chapter*{Declaration}
\addcontentsline{toc}{chapter}{Declaration}
\thispagestyle{plain}

I declare that this thesis is my own original work. Where I have drawn on the ideas, data, or words of others, I have acknowledged this through citation. This work has not been submitted for any other degree or qualification.

\vspace{2.5cm}
\noindent\textbf{Sergio Montenegro} \hfill April 2026
\clearpage

% ── Acknowledgements ──────────────────────────────────────────────────────────
\chapter*{Acknowledgements}
\addcontentsline{toc}{chapter}{Acknowledgements}
\thispagestyle{plain}

This research is motivated by the practical challenges facing contemporary democratic institutions, and by the existence of established bodies of democratic theory, institutional design, and cryptographic engineering that---while rarely combined---offer complementary tools for addressing them. The intellectual debts run wide: from Elinor Ostrom's patient institutional empiricism and James Fishkin's deliberative polling innovation, to the empirical study of Swiss direct democracy and the architectural work emerging from decentralised computing research.

I am grateful to the open-source research community whose tools make transparent, reproducible scholarship possible.

\clearpage

% ── Abstract ──────────────────────────────────────────────────────────────────
\chapter*{Abstract}
\addcontentsline{toc}{chapter}{Abstract}
\thispagestyle{plain}

\noindent\textbf{Background and motivation.}
Public confidence in democratic institutions has declined measurably across OECD countries over the past three decades \citep{norris2011democratic, dalton2004democratic}. The International IDEA Global State of Democracy report \citep{idea2021global} recorded that by 2021 more countries were experiencing democratic backsliding than democratic progress. Voter turnout in national elections has fallen in most established democracies; non-electoral participation is systematically skewed by income, education, and age \citep{schlozman2012unheavenly}. Decision processes remain insufficiently transparent: a comprehensive study of US policy outcomes by \citet{gilens2014testing} found that the preferences of median-income citizens had near-zero independent influence on policy once the preferences of economic elites and organised interest groups were controlled for. These empirical findings pose a concrete design problem: can governance institutions be restructured to broaden participation, increase transparency, and improve accountability without sacrificing stability or decision quality?

\vspace{0.7em}
\noindent\textbf{Scope and objectives.}
This thesis proposes \emph{Programmable Participatory Governance} (PPG), a formal framework designed to address the empirically documented deficits above. PPG synthesises five research traditions: (1)~deliberative democratic theory, drawing on the work of Habermas, Fishkin, and Gutmann and Thompson; (2)~participatory and direct democracy, drawing on Pateman, Barber, and the empirical record of Swiss direct democracy \citep{linder2010swiss, frey1994direct}; (3)~institutional economics and collective action theory, principally Ostrom's design principles for self-governing institutions \citep{ostrom1990governing}; (4)~cryptographic and distributed systems research enabling trustless, verifiable computation \citep{buterin2014ethereum, boneh2020graduate}; and (5)~political economy and institutional critique, drawing on the convergent analyses of Hayek, Sowell, Herman and Chomsky, Hirschman, and Scott on dispersed knowledge, structural media distortion, voice mechanisms, and administrative legibility failures. The objective is a framework that is simultaneously grounded in democratic theory, formally specified, empirically validated through simulation, and architecturally specified for practical deployment.

\vspace{0.7em}
\noindent\textbf{Framework.}
PPG structures governance as a sequential six-stage pipeline: proposal submission, structured deliberation, constitutional validation, verifiable voting, execution, and citizen-triggered veto. The pipeline is formalised as a state machine $\mathcal{G} = (S, s_0, \Sigma, \delta, F)$, with all transitions being deterministic and auditable. A key design contribution is the dynamic quorum mechanism: rather than fixing a single participation threshold, the required quorum scales with the estimated impact of each proposal ($Q(i) = Q_{\text{base}} + \alpha\cdot\sigma(i)$), calibrating the legitimacy requirement to the decision's stakes. This adapts the practice observed in Swiss constitutional referendums---where double-majority requirements apply only to constitutional changes---to a general, continuously parameterisable setting.

A second contribution is a composite legitimacy score $L(d)$ weighting participation rate, approval margin, and deliberation quality, enabling cross-decision comparison and longitudinal tracking of institutional health. A third contribution is the \emph{Transparency Default}: drawing on open government research \citep{oecd2020open} and Arendt's analysis of the public realm \citep{arendt1958human}, the framework specifies that all governance actions are public records by default and that no citizen surveillance infrastructure is architecturally possible at the protocol level. This is operationalised in the system architecture through strict data minimisation and zero-knowledge cryptographic proofs of eligibility.

\vspace{0.7em}
\noindent\textbf{Identity model.}
Digital governance is uniquely vulnerable to Sybil attacks---the creation of multiple false identities to inflate voting weight \citep{douceur2002sybil}. PPG proposes a hybrid identity model combining: social-graph verification \citep{pentland2014social}, which uses the structure of mutual-attestation networks to identify likely fake accounts; zero-knowledge proofs of eligibility \citep{boneh2020graduate}, which allow citizens to prove residence and age without revealing personal data; and non-transferable Soulbound Token credentials \citep{buterin2022desoc}, which prevent the buying and selling of governance standing. No existing solution fully satisfies all four desired properties (uniqueness, privacy, scalability, Sybil resistance) simultaneously; the proposed hybrid is a current best-practice synthesis that the thesis explicitly treats as subject to revision.

\vspace{0.7em}
\noindent\textbf{Simulation evidence.}
Stochastic simulations across three behavioural archetypes (passive, active, and strategic agents) over 100 decision rounds demonstrate the existence of parameter regimes in which several framework properties hold: the dynamic quorum mechanism maintains system throughput even as legitimacy requirements increase; and the veto mechanism operates as a stability-preserving safety valve rather than a destabilising reversal mechanism. Game-theoretic analysis (Chapter~\ref{sec:gametheory}) further demonstrates that the cost of mobilising a sufficiently large faction to manipulate outcomes exceeds the expected gain under a broad range of quorum levels.

\vspace{0.7em}
\noindent\textbf{Decentralised architecture.}
The PPG framework is designed to be deployable on centralised, federated, or decentralised infrastructure. A five-layer decentralised reference architecture is specified---settlement layer, identity layer, execution layer, coordination layer, and interface layer---drawing on Ethereum-based smart contract technology, Layer 2 scaling networks, content-addressed storage, and progressive web application design. This architecture addresses the trust problem that has undermined previous digital democracy platforms: by encoding governance rules in publicly auditable, self-enforcing smart contracts, the design eliminates reliance on any single operator. Lessons from existing DAO governance models identify the key structural failure of current on-chain governance (plutocratic one-token-one-vote) and show how PPG's identity-grounded design corrects it.

\vspace{0.7em}
\noindent\textbf{Iterative design, limitations, and implementation pathway.}
The framework is explicitly iterative: its own parameters are governed through the same participatory pipeline it describes, following Ostrom's principle that effective institutions evolve through practice \citep{ostrom1990governing}. The thesis identifies eight significant limitations: the identity problem is partially rather than fully solved; the simulations are stylised stochastic models (not full agent-based models) and are simplifications of real population heterogeneity; game-theoretic model assumptions introduce conservative biases that require empirical correction; deliberation quality is difficult to guarantee at scale; constitutional and legal integration is complex; digital access inequality remains a structural barrier; deliberation is subject to capture by well-resourced actors; and decentralised architecture components are at varying stages of maturity. These limitations are not minimised but addressed through a three-tier implementation pathway (local, regional, national) that maps each limitation to a concrete resolution strategy calibrated to the stakes and institutional context of each deployment scale.

\vspace{0.7em}
\noindent\textbf{Keywords:} participatory governance, deliberative democracy, direct democracy, Swiss direct democracy, formal models, state machine, dynamic quorum, stochastic simulation, identity, Sybil resistance, zero-knowledge proofs, Soulbound Tokens, game theory, decentralised autonomous organisations, smart contracts, open government, transparency, financial transparency ledger, civic participation mandate, constrained administrative authority, merit-based accountability, implementation pathway, institutional design.

\clearpage

% ── Table of Contents, Figures, Tables ────────────────────────────────────────
\tableofcontents
\clearpage
\listoffigures
\clearpage
\listoftables
\clearpage

% ══════════════════════════════════════════════════════════════════
\chapter{Introduction}
\label{sec:intro}

Governance institutions are the mechanisms through which collective decisions are made, enforced, and revised. Their performance can be assessed along several empirically measurable dimensions: participation breadth, decision transparency, accountability responsiveness, and resistance to capture by organised minorities. On all four dimensions, the available evidence suggests that contemporary democratic institutions are under significant strain.

\section{Empirical Context}

Public confidence in national governments fell in the majority of OECD countries between 1995 and 2020 \citep{norris2011democratic}. The International IDEA Global State of Democracy report for 2021 recorded that the number of countries experiencing democratic backsliding exceeded those showing democratic progress for the fifth consecutive year \citep{idea2021global}. Voter turnout in national elections has declined in most advanced democracies over a 30-year period \citep{dalton2004democratic}. Non-electoral participation---engagement with public consultations, petitions, local government---is measurably unequal: \citet{schlozman2012unheavenly} document systematic income-, education-, and age-based stratification in participatory voice across the United States. \citet{gilens2014testing}, analysing 1{,}779 US policy outcomes, find that the preferences of median-income citizens have near-zero independent predictive power on policy once elite and interest-group preferences are controlled for---a finding consistent with Dahl's warning about the structural conditions under which polyarchies degenerate into oligarchies \citep{dahl1989democracy}.

Decision-making opacity compounds these participation deficits. The OECD's open government surveys consistently find that citizens report insufficient access to government information and insufficient opportunities to influence policy before it is made \citep{oecd2020open}. Accountability operates primarily through electoral cycles that, in most systems, span four or five years---a temporal mismatch that allows consequential decisions to take effect with limited real-time citizen response.

These observations do not constitute an argument against representative democracy or for any particular reform ideology. They are measurable performance gaps that motivate an engineering question: \emph{what institutional design changes could improve performance on these dimensions, and at what costs?}

\section{Research Question and Approach}

This thesis addresses the following research question:

\begin{mdframed}
\textit{Can a formal governance framework---integrating structured deliberation, dynamic participation thresholds, verifiable voting, and citizen-triggered veto mechanisms---improve participation breadth, decision transparency, and accountability responsiveness in democratic systems, while maintaining stability and resistance to manipulation by coordinated minorities?}
\end{mdframed}

We propose \textbf{Programmable Participatory Governance} (PPG), a framework that operationalises several well-established theoretical contributions in institutional design---Ostrom's commons governance principles \citep{ostrom1990governing}, Fishkin's deliberative polling model \citep{fishkin1991democracy}, the Swiss direct democracy tradition \citep{linder2010swiss}---and extends them with a formal computational specification suitable for digital implementation.

The approach is comparative and incremental rather than revolutionary: PPG is calibrated against empirical precedents (particularly Switzerland), formally specified to enable verification, and validated through stochastic simulation. It is not presented as a replacement for existing institutions but as a design framework from which components can be selectively adopted.

\section{Contributions}

The specific contributions of this thesis are:
\begin{enumerate}
  \item A formal state-machine model of the governance pipeline, enabling deterministic auditing of all state transitions.
  \item A dynamic quorum mechanism calibrating participation requirements to proposal impact, extending the Swiss double-majority precedent to a continuously parameterisable model.
  \item A composite legitimacy score for quantitative, cross-decision governance evaluation.
  \item A hybrid identity model balancing Sybil resistance, privacy, and scalability, grounded in current cryptographic best practice.
  \item Stochastic simulation evidence of system resilience across behavioural archetypes, including adversarial agents.
  \item Game-theoretic analysis of manipulation resistance: a single-round deterrence proof (Theorem~\ref{thm:manipulation}) and a repeated-game collusion containment result (Proposition~\ref{prop:collusion}) establishing the conditions under which dynamic quorum parameters deter both single-round manipulation and sustained long-run capture.
  \item A Transparency Default specification---all governance actions public by default; citizen data private by default---grounded in open government research \citep{oecd2020open} and operationalised architecturally.
  \item A Financial Transparency Ledger specification requiring real-time, transaction-level, machine-readable public disclosure of all public financial flows.
  \item A Civic Participation Mandate module treating governance engagement as a civic obligation analogous to tax filing, with blank-vote provisions ensuring no political expression is compelled.
  \item A Constrained Administrative Authority specification binding public bodies to explicitly mandated scope and requiring public deliberation records.
  \item A Merit-and-Mandate public role accountability model applying dual accountability---professional competence and citizen mandate---to all public positions.
  \item A five-layer decentralised reference architecture enabling trustless, operator-independent deployment.
  \item An iterative governance framework applying Ostrom's adaptive institution principles to protocol evolution.
  \item A three-tier implementation pathway (local, regional, national) with concrete strategies for resolving each identified limitation at each scale.
\end{enumerate}

\section{Structure}

Chapter~\ref{sec:related} reviews five primary research traditions---deliberative democratic theory, participatory and direct democracy, collective action theory, cryptographic and distributed systems research, and political economy and institutional critique---alongside the game-theoretic and mechanism design foundations used in Chapters~\ref{sec:gametheory} and~\ref{sec:properties}. Chapter~\ref{sec:problem} analyses the documented problem space with empirical evidence and identifies the structural logic common to the five diagnosed deficits (Section~\ref{sec:conversion}). Chapter~\ref{sec:framework} presents the governance framework, including the Transparency Default, Financial Transparency Ledger, Civic Participation Mandate, Constrained Administrative Authority, and Merit-and-Mandate model. Chapter~\ref{sec:formal} develops the formal model. Chapter~\ref{sec:identity} addresses the identity problem. Chapter~\ref{sec:architecture} describes the system architecture. Chapter~\ref{sec:simulation} presents simulation methodology and results. Chapter~\ref{sec:gametheory} provides game-theoretic analysis, including a single-round manipulation deterrence proof and a repeated-game collusion containment result. Chapter~\ref{sec:properties} analyses system properties. Chapter~\ref{sec:decentralised} specifies the decentralised architecture. Chapter~\ref{sec:limitations} identifies limitations. Chapter~\ref{sec:pathway} develops a three-tier implementation pathway from local to national deployment. Chapter~\ref{sec:conclusion} concludes.

% ══════════════════════════════════════════════════════════════════
\chapter{Related Work}
\label{sec:related}

PPG synthesises contributions from five primary research traditions, reviewed across eight thematic sections; the section on game theory and social choice additionally provides the formal foundations for Chapters~\ref{sec:gametheory} and~\ref{sec:properties}. Each section identifies the empirical findings and theoretical tools that inform specific design choices.

\section{Democratic Theory and Legitimacy}

\citet{dahl1989democracy} provides the foundational criteria for democratic legitimacy used throughout this thesis: inclusion of all affected parties, equal consideration of interests, enlightened understanding (access to relevant information), effective participation, and agenda control. His polyarchy framework is explicitly empirical---it describes institutional approximations of democratic ideals and acknowledges the structural tensions inherent in large-scale democracies. \citet{lupia2006dahl} extends Dahl's framework to modern information environments, noting that democratic legitimacy in the twenty-first century depends critically on whether citizens have access to the information needed to make competent decisions.

\citet{habermas1996between} provides a procedural theory of democratic legitimacy grounded in communicative rationality: decisions are legitimate insofar as they emerge from rational discourse free of strategic distortion. This is an idealist standard that real deliberation never fully meets, but it provides the theoretical basis for evaluating deliberation quality---the design criterion operationalised in PPG's legitimacy score $\Delta_{\text{delib}}$.

\citet{rawls1971theory} contributes the principle of fair procedure as a foundation for just institutions, independent of the substantive outcomes produced. The veil-of-ignorance device motivates PPG's structural equality constraints: governance rules are evaluated as if participants do not know which position they will occupy within the system.

\section{Deliberative Democracy: Empirical Evidence}

The deliberative democracy tradition moves from normative theory to the empirical question of whether structured deliberation actually changes preferences and improves decision quality. \citet{fishkin1991democracy} introduces deliberative polling as a method for measuring considered public opinion: random samples of citizens are given balanced briefing materials, engage in moderated small-group discussion with expert panels, and are surveyed before and after. The empirical findings are consistent across more than 100 deliberative polls conducted in 30 countries \citep{fishkin2018stanford}: deliberation systematically increases knowledge, moderates extreme positions, and shifts preferences toward more other-regarding stances. Importantly, the changes are durable rather than ephemeral. \citet{dutwin2003character} provides complementary experimental evidence that structured group discussion increases argument quality and factual accuracy.

\citet{gutmann2004deliberating} provide a political philosophy grounding for deliberative institutions, arguing that the moral case for deliberative democracy rests on three grounds: it tends to produce epistemically better decisions (because more information is considered), it promotes mutual respect across disagreement, and it generates decisions with higher legitimacy because participants have been heard. \citet{gastil2010promise} reviews the empirical literature on deliberative practices in real policy contexts, finding consistent positive effects on participant knowledge, civic engagement, and willingness to accept decisions they did not personally favour---a finding directly relevant to PPG's veto threshold design.

\section{Participatory and Direct Democracy}

\citet{pateman1970participation} provides the developmental argument for participation: the act of participating in governance cultivates the civic competencies and dispositions needed to sustain self-governance over time. This is an empirical claim, not merely a normative one, and it is supported by longitudinal civic engagement research showing that early participatory experiences predict higher long-term political engagement \citep{held2006models}.

\citet{barber1984strong} distinguishes thin democracy---periodic voting as the primary participatory act---from strong democracy, in which citizens engage continuously in the negotiation of shared public life. This distinction is empirically grounded in the observation that electoral participation alone produces systematically less responsive policy than systems with between-election participatory mechanisms \citep{fung2006varieties}. \citet{fung2006varieties} provides an analytical taxonomy of participatory governance institutions across three dimensions---who participates, how they communicate and decide, and how deliberation connects to power---that PPG uses to map its design choices.

\section{The Swiss Model of Direct Democracy}

Switzerland provides the most extensively studied national system of direct democracy, with continuous institutional evolution since the federal constitution of 1848. The three main mechanisms---mandatory referendums on constitutional amendments, optional referendums on legislation (50{,}000 signatures within 100 days), and popular initiatives for constitutional proposals (100{,}000 signatures within 18 months)---create ongoing, between-election channels for citizen agenda-setting and legislative veto.

The empirical record supports several favourable effects. \citet{frey1994direct} finds measurably higher reported life satisfaction in Swiss cantons with greater direct democratic instruments, even after controlling for income and other confounders, interpreted as reflecting the value of procedural self-determination beyond policy outcomes. \citet{stutzer2000happiness} extends this finding to happiness measures across Swiss municipalities. \citet{linder2010swiss} provides a comprehensive assessment, documenting higher voter knowledge about policy details, greater policy responsiveness to median voter preferences, and lower political alienation compared to purely representative systems.

Critical assessments are equally important. Deliberation before referendums is largely informal and subject to unequal campaign resource effects. Linguistic minorities have been outvoted on cultural issues. Turnout in referendums is typically 40--50\%---high by international comparison but still leaving a majority non-participating \citep{linder2010swiss}. The fixed-rules structure provides no adaptive calibration mechanism.

PPG takes the Swiss model as its primary empirical reference point: it preserves the core mechanisms (between-election voting, citizen-initiated proposals, minority veto) while addressing the identified weaknesses through formal deliberation architecture, dynamic quorum calibration, and digital infrastructure that removes signature-collection barriers.

\section{Collective Action and Institutional Design}

\citet{ostrom1990governing} establishes through comparative case analysis that communities can effectively self-govern shared resources without either state control or privatisation, provided institutional design satisfies a set of empirically derived principles: clearly defined boundaries, rules proportional to local conditions, collective choice arrangements, monitoring, graduated sanctions, conflict resolution, and recognised rights to organise. Her analysis of why some commons institutions succeed and others fail is directly applicable to governance systems: the failure modes---capture by subgroups, inadequate monitoring, inability to adapt rules---are precisely the failure modes PPG is designed to resist.

\section{Game Theory and Social Choice}

\citet{arrow1951social} proves that no social welfare function satisfying four minimal conditions simultaneously (unrestricted domain, Pareto efficiency, independence of irrelevant alternatives, non-dictatorship) exists. Arrow's theorem is foundational for PPG's design: it establishes that any governance mechanism involves trade-offs, ruling out the possibility of a universally optimal procedure. PPG's response is to make the trade-offs explicit and configurable rather than hidden.

\citet{nash1950equilibrium} provides the equilibrium concept used in Chapter~\ref{sec:gametheory} to characterise strategic behaviour within the governance system. The repeated-game extension developed in Chapter~\ref{sec:gametheory} draws on the folk theorem for infinitely repeated games \citep{fudenberg1986folk}, which characterises the conditions under which long-run collusion can be sustained as a subgame-perfect equilibrium. The mechanism design results of \citet{hurwicz1973design} and \citet{maskin2008mechanism} provide the formal basis for the epistemic argument formalised in Chapter~\ref{sec:properties}.

\section{Digital Governance and Decentralised Systems}

The emergence of programmable blockchains \citep{buterin2014ethereum} created the first infrastructure on which governance rules could be expressed as self-enforcing code: proposals, votes, and executions are recorded immutably, quorum conditions are enforced by contract, and no operator can alter recorded outcomes. \citet{defilippi2018blockchain} analyse the jurisprudential implications of this shift.

Decentralised autonomous organisations (DAOs) are the first large-scale experiments with on-chain governance \citep{hassan2021decentralized}. The empirical record is instructive in both directions: DAOs demonstrate technical feasibility but consistently exhibit plutocratic dynamics, low participation rates (typically below 10\% of token holders), and governance capture \citep{chohan2017dao, siri2022daos}. \citet{reijers2021onchain} analyse the structural tension between on-chain verifiability and off-chain deliberation quality. These failure modes, and the design corrections PPG proposes, are analysed in detail in Chapter~\ref{sec:decentralised}.

\citet{buterin2022desoc} propose Soulbound Tokens as non-transferable identity credentials for Web3, directly relevant to PPG's identity model (Chapter~\ref{sec:identity}). \citet{ford2020identity} surveys the proof-of-personhood problem, establishing the technical requirements that PPG's hybrid identity model is designed to satisfy. \citet{mueller2019digital} provides a sceptical political science assessment of digital democracy platforms, cataloguing the failure modes of direct e-democracy implementations, which PPG's deliberation layer is designed to address.

Beyond DAO governance, a parallel tradition of civic technology platforms has developed for non-financial participatory decision-making at municipal and national scale. These are the closest direct precursors to PPG's governance pipeline, and their failure modes are directly instructive. Taiwan's vTaiwan process, initiated in 2014, combines the Pol.is deliberation platform---which uses machine learning to identify points of cross-partisan consensus from citizen statements---with structured government response commitments; over 26 national policy items have been adopted through the process \citep{tang2019digital}. Iceland's Better Reykjavik platform, deployed since 2010, enables citizens to submit and prioritise policy ideas, with over 2{,}500 ideas implemented at municipal level \citep{gilman2016better}. The Decidim platform, deployed in over 400 municipalities globally including Barcelona and Helsinki, provides open-source tools for participatory budgeting, policy proposals, and citizen assemblies, and represents the most institutionally developed civic technology commons currently in operation \citep{barandiaran2017decidim}.

None of these platforms integrates formal state machine governance, cryptographically anonymous voting, a verifiable financial ledger, or a dynamic quorum mechanism. Table~\ref{tab:civictech} places PPG against these platforms and the leading DAO governance systems on four properties that the analysis of failure modes identifies as necessary for trustworthy large-scale participatory governance.

\begin{table}[htbp]
\centering
\caption{Comparison of PPG with existing civic technology and DAO governance platforms.}
\label{tab:civictech}
\begin{tabular}{lcccc}
\toprule
\textbf{Platform} & \textbf{Formal} & \textbf{Anonymous} & \textbf{Financial} & \textbf{Dynamic} \\
& \textbf{state machine} & \textbf{voting} & \textbf{ledger} & \textbf{quorum} \\
\midrule
vTaiwan / Pol.is & No & No & No & No \\
Better Reykjavik & No & No & No & No \\
Decidim & Partial & No & No & No \\
Compound (DAO) & Partial & No & No & No \\
\textbf{PPG (this work)} & \textbf{Yes} & \textbf{Yes} & \textbf{Yes} & \textbf{Yes} \\
\bottomrule
\end{tabular}
\end{table}

The systematic gap across all existing platforms is the simultaneous combination of formal verifiability, vote anonymity, and financial transparency. Individual features exist in isolation; PPG integrates them as architectural requirements rather than optional features.

\section{Political Economy and Institutional Critique}

PPG draws on a set of contributions from political economy and institutional critique that span conventional ideological divisions, reflecting the thesis's concern with structural analysis rather than partisan prescription. These thinkers arrive at convergent conclusions about the epistemic and accountability failures of governance from independent analytical traditions---a convergence that strengthens rather than coincidentally supports the framework's design commitments.

\citet{hayek1945use}, in one of the most widely cited articles in twentieth-century economics, argues that the central problem of social organisation is not the allocation of known resources by a planner but the utilisation of knowledge that is inherently dispersed, local, and tacit---knowledge that no central authority can aggregate. This epistemological argument applies directly to governance: the participatory mechanisms in PPG rest not only on normative grounds of democratic legitimacy but on the epistemic claim that dispersed citizen knowledge about local conditions, preferences, and consequences cannot be adequately substituted by centralised administrative discretion. \citet{sowell1980knowledge} extends this argument to a systematic study of institutional design, demonstrating how the structure of incentives and information flows determines whether dispersed knowledge is used or suppressed by organisational arrangements.

\citet{herman1988manufacturing} provide a structural analysis of how mass media filter the public agenda through five mechanisms: ownership concentration, advertising dependence, official sourcing, flak, and ideological assumptions. Their analysis is directly relevant to PPG's deliberation layer: the same structural distortions that shape media discourse---disproportionate amplification of well-resourced voices, marginalisation of counter-narratives, agenda-setting by organised interests---operate in any public deliberation process without architectural correctives. The argument mapping, balanced briefing, and public archive specifications in PPG's deliberation layer are designed, in part, to resist precisely these dynamics.

\citet{hirschman1970exit} introduces the conceptual pair of \emph{exit} and \emph{voice} as the two fundamental mechanisms through which members of organisations signal dissatisfaction: they can leave (exit) or they can speak (voice). His analysis demonstrates that where exit is cheap, voice atrophies---members who can disengage have reduced incentive to organise for reform. Applied to democratic governance: citizens who disengage (abstention as exit) degrade the quality of the collective participatory signal. PPG's Civic Participation Mandate (Section~\ref{sec:mandate}) is grounded in Hirschman's logic: by raising the cost of disengagement and providing structured voice mechanisms, the system maintains the integrity of participatory input.

\citet{scott1998seeing} analyses the epistemological failures of twentieth-century state-led planning, showing how administrative systems systematically simplify and codify the social world to make it ``legible'' for management---and how this legibility project destroys the local, practical knowledge (\textit{m\^etis}) on which functioning social systems depend. Scott's critique of unconstrained administrative discretion provides independent theoretical grounding for two of PPG's core design specifications: Constrained Administrative Authority (Section~\ref{sec:constrained}), which binds public bodies to explicitly mandated scopes; and the transparent deliberation record, which makes administrative reasoning visible and therefore contestable.

These contributions do not share a political tradition---Hayek and Sowell write from within classical liberalism and conservative political economy; Herman and Chomsky from structural media critique; Hirschman from institutionalist political economy; Scott from anthropology informed by libertarian socialism. Their convergence on the structural importance of dispersed knowledge, transparent deliberation, and constrained administrative power reflects not ideological agreement but independent analytical conclusions about how governance institutions characteristically fail and what design responses are indicated.

% ══════════════════════════════════════════════════════════════════
\chapter{Problem Analysis}
\label{sec:problem}

This chapter documents the specific, measurable deficits that PPG is designed to address. Each deficit is grounded in empirical evidence, and each motivates specific design choices in the framework.

\section{Participation Gap}

Citizen participation in governance is episodic and unequal. In OECD countries, voting in national elections is the primary participatory act for the majority of citizens; participation in between-election governance processes---public consultations, petitions, local government meetings---involves a much smaller and systematically unrepresentative subset of the population \citep{schlozman2012unheavenly}. The socioeconomic stratification of participation is well-documented: higher income, higher education, and older age are the strongest predictors of participatory engagement across virtually all established democracies \citep{held2006models, dalton2004democratic}.

This participation inequality has policy consequences. \citet{gilens2014testing} find that US federal policy outcomes correlate strongly with the preferences of the top income quintile and with organised interest group positions, but show near-zero correlation with median-income preferences. While this specific finding applies to the United States, comparable patterns of interest-group influence over policy have been documented in European contexts \citep{norris2011democratic}.

The participation problem is not simply one of low engagement: it is one of \emph{structured inequality} in who participates and whose voice is weighted. Any governance reform that broadens participation without addressing its distributional structure risks amplifying existing inequalities.

\begin{figure}[htbp]
\centering
\includegraphics[width=0.85\textwidth]{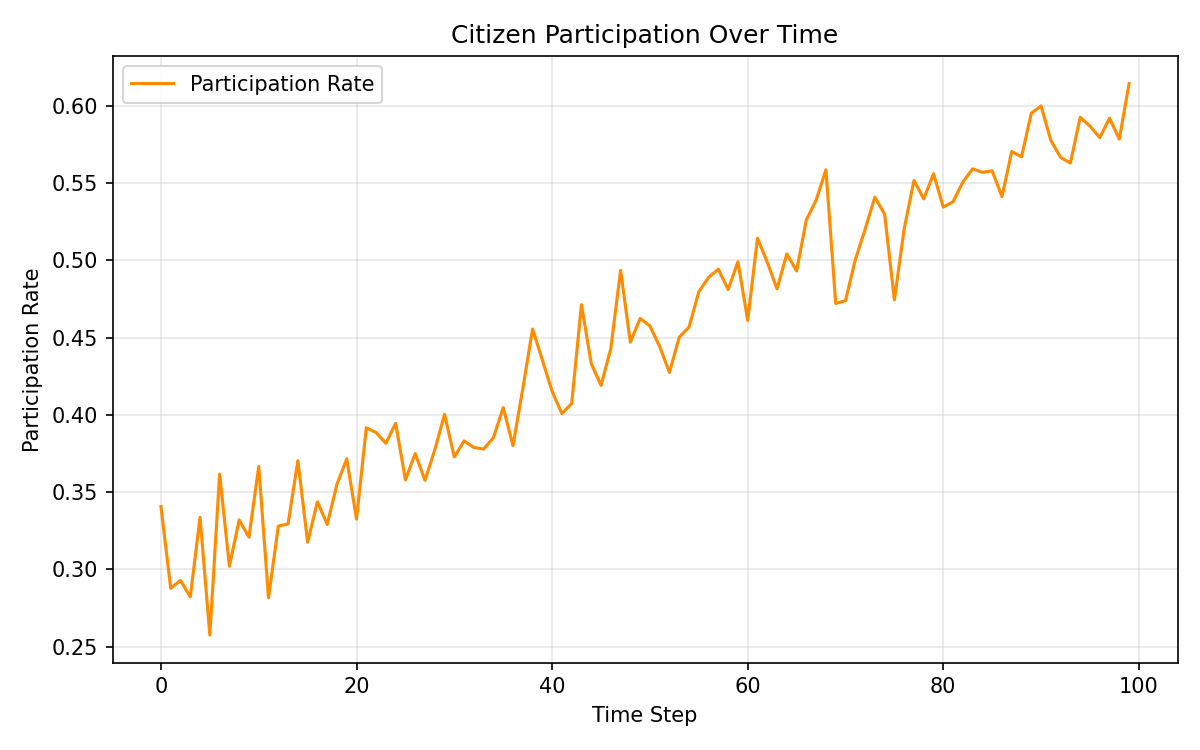}
\caption{Simulated citizen participation rates over 100 decision rounds under the PPG adaptive quorum mechanism. Participation gradually converges toward equilibrium as agents learn participation norms. The stochastic simulation incorporates three participation archetypes reflecting observed heterogeneity in civic engagement (see Chapter~\ref{sec:simulation}).}
\label{fig:participation}
\end{figure}

\section{Transparency Deficit}

Transparency in governance refers to the availability of relevant information---about deliberations, decisions, rationale, and implementation---to affected citizens. Transparency is both intrinsically valued (as a component of democratic legitimacy) and instrumentally important (as a precondition for accountability).

The OECD Open Government reviews consistently find that citizens across member states report insufficient access to information about how decisions are made \citep{oecd2020open}. Freedom of Information legislation, while expanding in scope, places the burden of information access on the requesting citizen rather than on the producing institution: documents are private unless released, rather than public unless restricted. In practice this creates substantial barriers: requests are delayed, partially redacted, or denied on broad exemption grounds. The gap between legal rights to information and practical access is large in most jurisdictions.

Transparency also affects deliberation quality. Where decision rationale is not documented, accountability is retrospective and impressionistic. Where legislative deliberations are not recorded accessibly, citizens cannot evaluate whether their interests were considered. PPG addresses this structurally: deliberation records, vote tallies, and decision rationale are archived as public-by-default records, reversing the access burden.

\section{Accountability Delay}

Representative democracy's primary accountability mechanism is electoral: citizens evaluate incumbent performance and vote accordingly at regular intervals. This creates a structural temporal mismatch. In most systems, election cycles span four to five years. Consequential decisions---fiscal commitments, infrastructure projects, regulatory changes---take effect immediately and may be difficult or impossible to reverse by the time the next election arrives.

The academic literature on democratic accountability consistently identifies this lag as a structural weakness rather than an incidental feature \citep{norris2011democratic}. Real-time accountability mechanisms---legislative oversight, judicial review, independent audit bodies---partially compensate but operate with limited citizen involvement and on slow institutional timescales. PPG's veto mechanism is designed as a between-election accountability instrument: citizens can trigger a review of a specific decision within a defined window, without requiring a general election.

\section{Influence Concentration}

Formal democratic equality (one person, one vote) coexists with substantial inequality in effective political influence. Well-resourced actors---business associations, professional lobbies, large individual donors---have structural advantages in agenda-setting, policy framing, and access to decision-makers that are not available to unorganised citizens \citep{schlozman2012unheavenly}. This is not a unique pathology of any one system: cross-national research finds that organised interests are disproportionately represented in policy processes relative to diffuse citizen interests in virtually all liberal democracies \citep{norris2011democratic}. \citet{herman1988manufacturing} provide a structural account of the mechanisms through which agenda-setting power is concentrated: ownership patterns, sourcing conventions, and resource asymmetries systematically filter whose concerns enter deliberative space, regardless of formal equal access provisions.

PPG does not eliminate organised interest influence, which would be both impractical and potentially counter-productive (organised expertise can be valuable in deliberation). It addresses influence concentration through two mechanisms: public deliberation records, which make lobbying activity and its effects visible; and mandatory deliberation phases with structured argument review, which require proposals to withstand scrutiny from a diverse range of participants rather than passing through informal channels.

\section{Transparency Asymmetry Between Institutions and Citizens}

A fifth structural feature, observable across governance systems, is an asymmetry in the direction of information flows: states collect substantial information about citizens (tax records, property registrations, identity documents, and, increasingly, digital behavioural data through surveillance and platform partnerships), while citizen access to information about state deliberations and operations is comparatively limited \citep{zuboff2019surveillance}.

This is an empirically observable feature of contemporary governance architecture, documented extensively in the open government literature \citep{oecd2020open} and in surveillance studies \citep{zuboff2019surveillance}. \citet{rosanvallon2008counter} identifies counter-democratic monitoring---organised citizen surveillance of the state---as the mechanism by which civil society partially compensates for this asymmetry, through investigative journalism, audit institutions, and civil society watchdog organisations.

PPG addresses this asymmetry at the architectural level. The Transparency Default (Chapter~\ref{sec:framework}) reverses the access burden on governance action data: it is produced as a public record by the system's infrastructure, requiring no petitioning. The ZK proof architecture (Chapter~\ref{sec:architecture}) addresses the \emph{direction} of the asymmetry: by making citizen surveillance structurally impossible at the protocol level, it inverts the default from `state watches citizen' to `citizen watches state'. No policy decision is required to maintain this inversion; it is enforced by the cryptographic protocol itself.

\section{The Structural Logic of Conversion}
\label{sec:conversion}

The five deficits documented in this chapter share a common structural origin. Participation inequality, transparency gaps, accountability lag, influence concentration, and institutional opacity are not independent failures requiring independent fixes: they are jointly produced by a single architectural feature of representative democracy as it has been historically implemented. That feature is the \emph{mediation of citizen preferences through political parties}.

Political parties solve a genuine coordination problem. In a polity of millions, individual citizens cannot practically deliberate on every decision, evaluate every candidate, or monitor every administrative action in real time. Parties aggregate dispersed preferences into coherent platforms, recruit and vet candidates for office, and organise legislative activity into stable governing blocs \citep{downs1957economic}. These are genuine coordination functions, and the durability of party-based democracy across diverse institutional settings reflects their practical necessity---in the absence of any alternative aggregation mechanism.

The aggregation mechanism PPG provides renders the party's core epistemic function redundant. When every citizen can enter a proposal into a structured deliberative pipeline, engage directly with argument mapping and expert consultation, and cast a cryptographically verifiable vote on the refined output, the party is no longer the necessary intermediary between citizen preferences and policy outcomes. The pipeline performs preference aggregation continuously, transparently, and with symmetric access---properties that party-mediated systems structurally cannot provide, because party survival depends on differential access to agenda-setting, funding, and information. \citet{hayek1945use}'s epistemological argument about dispersed knowledge---cited earlier in the context of administrative design---applies here with equal force: a party platform is a centralised summary of preferences, and centralised summaries systematically lose the local, contextual, and minority knowledge that direct participatory processes retain. Arrow's impossibility theorem \citep{arrow1951social} establishes that no aggregation procedure can simultaneously satisfy all democratic desiderata; PPG's claim is not that the pipeline solves this problem but that direct preference aggregation under the same constraints loses less information than party-mediated summarisation.

This is not an argument that parties are malign. It is a structural observation: the historical conditions that made party-based preference aggregation the only viable model---mass polities without real-time communication infrastructure, without verifiable identity at scale, without cryptographic audit trails---no longer obtain. PPG operationalises what was previously only normatively desirable but technically infeasible: direct preference aggregation at civic scale.

\subsection{The Conversion Sequence}

The conversion is not a rupture but a progressive functional substitution, proceeding through the same three tiers as the deployment pathway (Chapter~\ref{sec:pathway}).

\textbf{At Tier 1}, existing institutional structures remain intact. Elected office holders continue in their roles. Parties continue to organise legislative activity. PPG handles bounded, specific decisions---participatory budgeting, local planning priorities, community resource allocation---within the scope permitted by existing administrative law. Office holders implement PPG decisions within their existing authority. The effect at Tier 1 is institutional augmentation, not replacement: a direct participation layer is added, and citizens begin developing governance competence in low-stakes contexts. The comparison with accountability is instructive: just as financial audit does not replace management but changes what management must answer for, a functioning PPG pipeline does not replace elected office holders but changes what they must answer for.

\textbf{At Tier 2}, the balance shifts. As the pipeline accumulates a legitimacy track record, the contrast between pipeline-legitimated decisions and party-driven decisions becomes publicly visible through the legitimacy score. Party affiliation becomes progressively less relevant to an office holder's accountability than their PPG legitimacy record. A minister whose decisions consistently score high participatory legitimacy is accountable to citizens directly, through the pipeline's audit trail, rather than primarily through party discipline. Simultaneously, the Merit-and-Mandate model (Section~\ref{sec:merit}) begins to apply to administrative and technical roles: positions previously filled by political appointment become open, competency-assessed appointments, defined by mandate scope and evaluated against measurable outcomes. The party's patronage function atrophies as professional accountability replaces political loyalty as the operative selection criterion.

\textbf{At Tier 3}, the constitutional settlement has shifted. Office holders remain: executive management of complex public systems requires dedicated professionals, as it does in any well-governed large organisation. What changes is the \emph{source and structure of their authority}. Authority no longer flows from party nomination and electoral plurality; it flows from two channels operating in parallel---demonstrated competence assessed through open, transparent process (the merit channel) and citizen mandate expressed through the PPG pipeline (the mandate channel). The role of an official at Tier 3 is structurally analogous to that of a senior executive in a stakeholder-governed organisation: they exercise professional judgment and implement decisions within a mandate set by the governing body---in this case, the citizen pipeline---and are accountable to two independent performance channels simultaneously. Office is a professional role like any other, with open selection, public performance records, and removal through structured process rather than through electoral cycles and party loyalty.

Parties that adapt to this architecture become what they are nominally claimed to be in liberal democratic theory: advocacy organisations that enter proposals into the pipeline, organise deliberative contributions, and compete on argument quality and evidence rather than on organisational machinery and campaign financing. Parties that do not adapt lose their functional rationale.

\subsection{The Self-Referential Foundation}

The architecture is self-evolving, and this property is the key structural response to the deepest objection to any governance reform proposal: \emph{who decides when the reform is working, and what comes next?}

In PPG, the answer is the same as the answer to every other governance question: the citizens, through the pipeline, with full deliberative record and verifiable vote. Because PPG's own parameters---quorum thresholds, deliberation periods, role definitions, mandate scopes---are themselves governed through the PPG pipeline (Chapter~\ref{sec:properties}), the framework does not require an external reform process to develop further. Citizens who observe that a particular role definition is too broad, that a quorum threshold is miscalibrated, or that a Merit-and-Mandate criterion is poorly specified can enter that observation as a proposal into the same pipeline that governs all other decisions. The system refines itself through the same mechanism it uses to govern everything else.

This self-referential property distinguishes PPG from historical governance reform proposals, which typically require a separate constitutional moment---a founding act of authority from outside the system---to initiate and sustain reform. PPG requires a founding act only to establish Tier 1; thereafter, the system's own iterative governance mechanism drives its evolution, calibrated continuously against empirical participation and legitimacy data. The conversion is not decreed from outside; it emerges from within.

% ══════════════════════════════════════════════════════════════════
\chapter{Framework Overview}
\label{sec:framework}

PPG structures governance as a sequential, multi-layer pipeline. Each layer has defined entry criteria, processes, and outputs, making the system auditable and programmable. This chapter presents the pipeline architecture, then specifies five foundational design commitments that apply across the entire framework: the Transparency Default, the Civic Participation Mandate, the Financial Transparency Ledger, the Constrained Administrative Authority specification, and the Merit-and-Mandate public role accountability model.

\section{Governance Layers}

\begin{enumerate}
  \item \textbf{Proposal Layer.} Citizens or delegated representatives submit structured proposals. Proposals must include: a problem statement, proposed action, estimated impact scope, affected stakeholders, and resource requirements.

  \item \textbf{Deliberation Layer.} Proposals enter a timed deliberation period. Structured deliberation includes argument mapping, expert consultation, and impact assessment. The outcome is a refined proposal with documented rationale.

  \item \textbf{Validation Layer.} A validation step verifies that proposals comply with constitutional constraints, procedural requirements, and legal frameworks. This layer can be partially automated.

  \item \textbf{Voting Layer.} Validated proposals proceed to verifiable citizen voting. Participation requirements (quorum) scale dynamically with estimated impact.

  \item \textbf{Execution Layer.} Approved decisions are executed through binding implementation mechanisms, with progress tracking and public reporting.

  \item \textbf{Veto Layer.} Citizens may trigger a veto on executed decisions within a defined window. Veto requires a threshold level of support, providing a reversibility mechanism without enabling minority obstruction.
\end{enumerate}

\section{Pipeline Visualisation}

\begin{center}
\begin{tikzpicture}[node distance=1.6cm, auto,
  box/.style={draw, rectangle, minimum width=3.2cm, minimum height=0.7cm,
              fill=blue!8, font=\small},
  arrow/.style={->, thick, blue!60!black}]

  \node (p)  [box] {1.\ Proposal};
  \node (d)  [box, below of=p] {2.\ Deliberation};
  \node (v)  [box, below of=d] {3.\ Validation};
  \node (vt) [box, below of=v] {4.\ Voting};
  \node (e)  [box, below of=vt] {5.\ Execution};
  \node (ve) [box, below of=e] {6.\ Veto};

  \draw[arrow] (p) -- (d);
  \draw[arrow] (d) -- (v);
  \draw[arrow] (v) -- (vt);
  \draw[arrow] (vt) -- (e);
  \draw[arrow] (e) -- (ve);
  \draw[arrow] (ve.east) to[out=0,in=0, looseness=2.2]
               node[right, font=\scriptsize] {reversal} (p.east);
\end{tikzpicture}
\end{center}

\noindent\textit{\small The diagram shows the happy-path sequence only. Rejection transitions---triggered by proposer withdrawal during deliberation, failed constitutional validation, or insufficient quorum or approval margin at the voting stage---are omitted for visual clarity. The complete transition function $\delta$ is specified in Table~\ref{tab:transitions} in Chapter~\ref{sec:formal}.}

\section{The Transparency Default}
\label{sec:inversion}

Open government research consistently identifies default opacity as a structural barrier to democratic accountability \citep{oecd2020open}. The OECD Open Government Principles identify proactive disclosure---publishing government information without requiring citizen requests---as a best practice with demonstrated benefits for public trust, reduced corruption, and improved policy responsiveness. Countries with strong proactive disclosure regimes (Sweden, Finland, and New Zealand score consistently highest) show measurably higher civic trust and engagement \citep{oecd2020open}.

PPG adopts a \textbf{Transparency Default} as a core design constraint:

\begin{mdframed}
\textbf{Transparency Default:} \textit{All governance actions---proposals, deliberation records, votes, execution records, veto registrations, legitimacy scores---are produced as public data by the system infrastructure. No separate access request is required. Citizen personal data is minimised to what is required for eligibility verification and is protected by zero-knowledge cryptographic proofs.}
\end{mdframed}

This is a design specification, not a normative manifesto. It is operationalised in the system architecture (Chapter~\ref{sec:architecture}) through specific technical mechanisms: immutable public append-only logs for all governance state transitions; zero-knowledge proof protocols for identity verification that expose no personal data beyond eligibility; and content-addressed permanent archival of all deliberation records.

The privacy side of the constraint is equally specific. Drawing on Arendt's distinction between the public realm of collective action and the private realm of individual life \citep{arendt1958human}, and on current data minimisation principles in the privacy engineering literature \citep{boneh2020graduate}, PPG specifies that the governance infrastructure should not be capable of tracking individual participation patterns, voting histories, or personal identification. The zero-knowledge proof design achieves this: the system verifies that a participant is an eligible citizen without recording \emph{which} eligible citizen they are.

Operationally, the Transparency Default means:

\begin{itemize}
  \item All proposal texts, deliberation arguments, expert assessments, vote tallies, and execution records are public records, produced in machine-readable open formats.
  \item All legitimacy scores and governance metadata are publicly queryable via open API.
  \item No individual participation record, vote direction, or identity link is stored by any system component.
  \item The audit trail is append-only and cryptographically signed, preventing retroactive modification.
\end{itemize}

This design is calibrated against the OECD Open Government best practice standard \citep{oecd2020open} and against current Privacy by Design principles \citep{boneh2020graduate}, rather than derived from any particular political philosophy.

\section{Civic Participation Mandate}
\label{sec:mandate}

Most democratic systems treat participation in governance as a right: citizens may vote, propose, and deliberate, but are under no legal or institutional obligation to do so. The result---documented extensively in the participation inequality literature (Chapter~\ref{sec:problem})---is systematic under-participation by lower-income, lower-education, and younger citizens: precisely those whose interests are most weakly represented in existing governance structures \citep{lijphart1997unequal, schlozman2012unheavenly}.

PPG proposes a configurable \textbf{Civic Participation Mandate} as an optional design module, drawing on an analogy to the civic obligation to file tax returns. The analogy is instructive: tax filing is compulsory not because the state distrusts citizens, but because accurate collective information about income flows is a necessary input to managing shared fiscal infrastructure. Governance participation provides an analogous input: the aggregate of citizen preferences is the necessary signal on which legitimate collective decisions rest. Non-participation is not a neutral absence---it systematically skews that signal, with predictable distributional consequences \citep{lijphart1997unequal}.

The mandate operates as follows:

\begin{itemize}
  \item Eligible citizens are notified of each governance decision within their scope and are obligated to register their participation status within the voting window.
  \item \textbf{A blank or abstain vote constitutes full participation.} The obligation is to engage with the process---to register ``I have been informed of this decision and decline to direct it''---not to hold or express any particular opinion. Compelled opinion is not a democratic value; compelled engagement with the civic process is a civic contribution of the same kind as tax filing or jury service.
  \item Participation is verified through the zero-knowledge eligibility system: the system records that an eligible citizen participated, but not how they voted or whether their vote was blank.
  \item The specific enforcement mechanism is a framework parameter: jurisdictions may implement social notification (public participation records showing whether a citizen participated, but not how), minor financial incentive/disincentive structures, or voluntary systems with quorum-weighted legitimacy scoring. PPG specifies the mechanism; the deployment context determines the enforcement model.
\end{itemize}

Empirical evidence from compulsory voting systems supports the principle. \citet{lijphart1997unequal} analyses participation data across 36 democracies and finds that compulsory voting substantially reduces the socioeconomic participation gap: the differential between high-income and low-income participation rates narrows significantly where voting is mandatory. Australia has operated compulsory voting since 1924, consistently achieving turnout above 90\% \citep{idea2021global}. The blank vote provision is critical: it ensures the mandate does not compel political expression---a distinction confirmed in Australian constitutional jurisprudence.

The Civic Participation Mandate is itself a governance decision: its adoption within a PPG deployment is subject to the same participatory process it governs.

\section{Financial Transparency Ledger}
\label{sec:ledger}

The Transparency Default (Section~\ref{sec:inversion}) specifies that all governance actions are public records. The Financial Transparency Ledger extends this principle to its most consequential domain: the complete, real-time transparency of every flow of public funds.

\begin{mdframed}
\textbf{Financial Transparency Ledger:} \textit{Every transaction from public funds---salary payments, contract awards, expense claims, budget allocations, inter-departmental transfers, grant disbursements, and debt instruments---is recorded in real time as an immutable, machine-readable public entry. No minimum transaction threshold applies. No category is exempt from recording, though specific fields may be classified under narrowly defined and publicly logged security exceptions.}
\end{mdframed}

The current state of public financial transparency falls substantially short of this standard. Most jurisdictions publish aggregated budget documents and post-hoc annual reports; some publish contract award notices above threshold amounts; few publish transaction-level real-time data. Information is delayed, aggregated, inconsistently formatted, and frequently subject to broad exemptions \citep{oecd2020open}. The Open Contracting Data Standard \citep{ocp2019ocds}, adopted in over 50 countries, represents significant progress toward machine-readable public contracting data, but it covers procurement only---not the full spectrum of public financial flows.

The PPG Financial Transparency Ledger specifies:

\begin{itemize}
  \item \textbf{Transaction-level granularity.} Every individual payment recorded separately, not as aggregated budget categories.
  \item \textbf{Real-time publication.} Transactions published within 24 hours of execution, not in quarterly or annual reports.
  \item \textbf{Immutable audit trail.} All entries cryptographically signed at the time of recording. Corrections and reversals are recorded as new entries referencing the original---not as retroactive modifications.
  \item \textbf{Machine-readable open format.} All entries conform to a published open schema, enabling automated analysis, anomaly detection, and independent audit.
  \item \textbf{Classified transaction protocol.} Transactions exempt from full disclosure for operational security reasons are logged as redacted entries: the existence, approximate category, responsible department, and applicable legal exemption are recorded; only specific fields are withheld. The redaction itself is public.
  \item \textbf{Budget-to-actuals linkage.} Each transaction is linked to the approved budget line authorising it. Expenditure outside approved lines triggers an automatic flag in the public record, providing a direct basis for citizen veto petition.
\end{itemize}

The Financial Transparency Ledger does not require blockchain infrastructure: a cryptographically signed, append-only database maintained by an independent audit institution satisfies the same functional properties. In a decentralised PPG implementation (Chapter~\ref{sec:decentralised}), the ledger is an on-chain public record, providing the strongest available tamper-resistance guarantee. Crucially, this specification has no dependency on the rest of the PPG framework: it can be adopted as a standalone open-government measure by any jurisdiction, providing an incremental adoption path.

Systematic analysis of public financial data identifies transaction-level transparency as among the most effective anti-corruption mechanisms available \citep{oecd2020open}: it is preventive rather than detective, continuous rather than episodic, and scales without requiring additional institutional enforcement capacity. The Financial Transparency Ledger is the state's financial system operating in full public view---precisely the inversion of current default opacity.

\section{Constrained Administrative Authority}
\label{sec:constrained}

Contemporary governance systems delegate substantial discretionary authority to administrative agencies: departments interpret their mandates expansively, develop internal policies without explicit citizen authorisation, and exercise regulatory powers whose aggregate scope far exceeds what citizens approve in any single election \citep{norris2011democratic}. This administrative expansion is a structural accountability gap---it operates between elections, in institutional spaces that the veto mechanisms of Swiss-style direct democracy were not designed to reach.

PPG addresses this through the principle of \textbf{Constrained Administrative Authority}:

\begin{mdframed}
\textbf{Constrained Administrative Authority:} \textit{Public bodies exercise deliberative and policy-making powers only within the scope explicitly mandated by the citizen governance process. All internal deliberations are public records subject to the Transparency Default. All policy outputs are eligible for citizen veto review within the standard window.}
\end{mdframed}

This principle has three components:

\textbf{Mandate-bounded scope.} Each public body's authority is defined by the specific mandate approved through the governance pipeline. Scope expansion requires a new governance decision. The civil service role is to implement approved mandates competently---not to interpret and extend them autonomously. Professional judgment within mandate scope is not only permitted but required for effective implementation; the constraint is on the boundary between implementation and unsanctioned policy-making.

\textbf{Transparent deliberations.} Internal deliberations of public bodies---committee meetings, inter-departmental consultations, policy drafting processes---are public records subject to the same Transparency Default as the governance pipeline. A short publication delay (configurable, suggested maximum 72 hours) may apply for operational reasons, but the permanent record is public. This makes the reasoning behind administrative decisions visible to citizens and researchers, enabling identification of mandate drift and providing an evidence base for veto petitions.

The transparency of deliberations serves a mechanism function beyond accountability: it raises public awareness. Citizens who would not have initiated a veto in the absence of information may do so when internal deliberation records surface the reasoning behind a consequential decision. Transparency is not only a property of the system---it is the activation mechanism for civic engagement. \citet{scott1998seeing} observes that administrative systems rendered opaque and simplified for managerial convenience systematically destroy the practical knowledge that effective oversight requires; making deliberations public is the precondition for citizens exercising the \textit{m\^etis}---local, contextual knowledge---that formal processes cannot substitute.

\textbf{Veto eligibility.} Any policy decision taken by a public body within its mandate is eligible for citizen veto review through the standard mechanism (Chapter~\ref{sec:formal}). The veto threshold scales to the impact of the decision via the dynamic quorum formula, preventing minority abuse while ensuring that significant public concern can trigger review. This is the digital-scale equivalent of the Swiss optional referendum---a continuous between-election accountability mechanism extending to administrative, not only legislative, decisions.

\section{Merit-and-Mandate: Public Role Accountability}
\label{sec:merit}

The accountability of individuals in public roles is a persistent structural weakness in governance. Elected officials face accountability at defined intervals but are insulated between elections; appointed officials face weak accountability to the public and primarily to their political superiors. Neither model produces strong incentives for sustained competence or public service orientation \citep{rauch2000bureaucratic}. Merit-based bureaucratic structures are associated with significantly better governance outcomes \citep{rauch2000bureaucratic}, but existing merit models provide no mechanism for the public to remove incompetent officials outside the political chain of command.

PPG specifies a \textbf{Merit-and-Mandate} model applying dual accountability to all public roles:

\begin{mdframed}
\textbf{Merit-and-Mandate:} \textit{Public roles are filled on the basis of demonstrated competence, assessed by independent process. Role holders are subject to removal through two independent channels: professional competence review and citizen mandate. Both channels operate continuously, not episodically.}
\end{mdframed}

\textbf{Appointment by merit.} All public roles are filled through a transparent, competency-based selection process. Role definitions, selection criteria, candidate applications, assessment results, and appointment rationale are public records. Appointment panels are independent of political direction. Track record, qualifications, and demonstrated performance are the primary selection criteria.

\textbf{Performance transparency.} Each role holder has a public performance record: mandate scope, decisions taken, resources managed, and measurable outcomes. This record is part of the Financial Transparency Ledger and the governance pipeline's audit trail, updated continuously.

\textbf{Removal by citizen mandate.} Citizens may initiate a removal petition for any public role holder through the standard veto mechanism. The threshold is set higher than a standard policy veto (reflecting the disruption cost of personnel removal) but the mechanism is identical. Valid removal grounds are specified and constrained: documented incompetence, corruption, or sustained operation outside mandate. Political disagreement with the official's decisions is not a valid removal ground---this protection is enforced by an independent accountability tribunal.

\textbf{Protection from political interference.} Public officials are protected from removal for: political views, whistleblowing, correct application of law, or refusal to implement unlawful instructions. This protection runs to the public and the independent tribunal, not to the political executive. The result threads a specific needle: officials cannot be removed by political superiors for doing their jobs correctly, but they can be removed by the public for not doing their jobs at all.

The Merit-and-Mandate model draws on several existing institutional precedents: judicial retention elections (used in 21 US states and multiple Latin American jurisdictions), Swiss recall mechanisms (\textit{Abberufungsrecht}), and senior civil service performance accountability models in New Zealand and Australia. PPG generalises these precedents into a unified architecture applicable to all public roles regardless of jurisdiction.

The systemic significance of Merit-and-Mandate within the framework's conversion logic is analysed in Section~\ref{sec:conversion}: it is the mechanism through which the party patronage function atrophies at Tier 2 as professional accountability replaces political loyalty as the operative selection criterion. A formal principal--agent treatment of the dual accountability structure---modelling the triadic relationship between citizen principals, role-holder agents, and the independent accountability tribunal---is deferred to post-deployment analysis.

% ══════════════════════════════════════════════════════════════════
\chapter{Formal Model}
\label{sec:formal}

We now formalise the governance system.

\section{State Machine}

\begin{definition}[Governance State Machine]
A governance system is a tuple $\mathcal{G} = (S, s_0, \Sigma, \delta, F)$ where:
\begin{itemize}
  \item $S = \{\textit{Proposed}, \textit{Deliberating}, \textit{Validating}, \textit{Voting}, \textit{Executed}, \textit{Vetoed}, \textit{Rejected}\}$ is the set of states.
  \item $s_0 = \textit{Proposed}$ is the initial state.
  \item $\Sigma$ is the set of triggering events (citizen actions, timer expiry, threshold crossings).
  \item $\delta: S \times \Sigma \to S$ is the transition function.
  \item $F = \{\textit{Executed}, \textit{Rejected}\}$ is the set of absorbing (terminal) states. \textit{Vetoed} is a transient reset state that transitions unconditionally back to \textit{Proposed} for a new governance cycle and does not terminate the process. \textit{Executed} becomes absorbing only upon veto window expiry; prior to that, it may transition to \textit{Vetoed} if the veto threshold is reached (see Table~\ref{tab:transitions}).
\end{itemize}
\end{definition}

Table~\ref{tab:transitions} enumerates all transitions of $\delta$, giving an unambiguous implementation specification for each state change.

\begin{table}[htbp]
\centering
\caption{Complete state transition function $\delta: S \times \Sigma \to S$ for the PPG governance state machine.}
\label{tab:transitions}
\begin{tabular}{@{} l l l @{}}
\toprule
\textbf{Current state} & \textbf{Triggering event} $(\in \Sigma)$ & \textbf{Next state} \\
\midrule
$\textit{Proposed}$ & Submission accepted; deliberation period opens & $\textit{Deliberating}$ \\
$\textit{Deliberating}$ & Deliberation period expires; validation opens & $\textit{Validating}$ \\
$\textit{Deliberating}$ & Proposer withdrawal & $\textit{Rejected}$ \\
$\textit{Validating}$ & Constitutional compliance confirmed & $\textit{Voting}$ \\
$\textit{Validating}$ & Constitutional compliance failed & $\textit{Rejected}$ \\
$\textit{Voting}$ & $R \geq Q(i)$ and $V_{\text{for}}/V > 0.5$ & $\textit{Executed}$ \\
$\textit{Voting}$ & $R < Q(i)$ or $V_{\text{for}}/V \leq 0.5$ & $\textit{Rejected}$ \\
$\textit{Executed}$ & $V_{\text{veto}}/P \geq Q_{\text{veto}}$ within window $\tau$ & $\textit{Vetoed}$ \\
$\textit{Executed}$ & Veto window $\tau$ expires without threshold & $\textit{Executed}$ (terminal) \\
$\textit{Vetoed}$ & Revised resubmission & $\textit{Proposed}$ (new cycle) \\
$\textit{Rejected}$ & --- & $\textit{Rejected}$ (terminal) \\
\bottomrule
\end{tabular}
\end{table}

\section{Participation and Quorum}

Let:
\begin{align}
  P &= \text{Total eligible population} \\
  V &= \text{Total votes cast on a proposal} \\
  V_{\text{for}} &= \text{Votes in favour} \\
  R &= \frac{V}{P} \quad \text{(participation rate)} \\
  Q &= \text{Quorum threshold} \quad (0 < Q < 1)
\end{align}

A proposal passes if and only if:
\begin{equation}
  \label{eq:pass}
  R \geq Q \quad \text{and} \quad \frac{V_{\text{for}}}{V} > 0.5
\end{equation}

\section{Dynamic Quorum}

A fixed quorum is inadequate: low-impact decisions should not require the same mobilisation as high-impact ones. We define a dynamic quorum:

\begin{equation}
  \label{eq:dynamic_quorum}
  Q(i) = Q_{\text{base}} + \alpha \cdot \sigma(i)
\end{equation}

where:
\begin{itemize}
  \item $Q_{\text{base}} \in (0, 1)$ is the baseline quorum (e.g., 0.20).
  \item $\alpha \in [0, 1]$ is the impact sensitivity parameter, subject to the joint constraint $Q_{\text{base}} + \alpha \leq 1$ to ensure $Q(i) \leq 1$ for all~$i$.
  \item $\sigma(i) \in [0, 1]$ is the normalised impact score of proposal $i$, assessed during the validation layer as a weighted combination of three inputs: (i)~estimated financial magnitude, normalised against the governance context's budget baseline; (ii)~estimated population affected, normalised against the total eligible population; and (iii)~reversibility class, mapped from a three-level ordinal (low, medium, high) to $[0, 1]$. The weighting of these inputs is a governance parameter, configurable through the same participatory process as $Q_{\text{base}}$ and $\alpha$.
\end{itemize}

\begin{proposition}
\label{prop:quorum}
For any $\alpha, Q_{\text{base}} > 0$, the dynamic quorum $Q(i)$ satisfies $Q_{\text{base}} \leq Q(i) \leq Q_{\text{base}} + \alpha$, ensuring that baseline participation is always required and that high-impact decisions require greater participation.
\end{proposition}

\begin{proof}[Proof sketch]
Since $\sigma(i) \in [0,1]$ by definition, $Q(i) = Q_{\text{base}} + \alpha\cdot\sigma(i) \in [Q_{\text{base}},\, Q_{\text{base}}+\alpha]$ for all $\alpha, Q_{\text{base}} > 0$. The bounds are tight: $\sigma(i)=0$ yields $Q(i)=Q_{\text{base}}$; $\sigma(i)=1$ yields $Q(i)=Q_{\text{base}}+\alpha$. \qed
\end{proof}

The empirical relationship between quorum threshold and mean participation is illustrated by the simulation results in Figure~\ref{fig:quorum_sensitivity}a (Chapter~\ref{sec:simulation}).

\section{Veto Mechanism}

Following execution, a decision enters the veto window of duration $\tau$. During this window, a veto is triggered if:

\begin{equation}
  \label{eq:veto}
  \frac{V_{\text{veto}}}{P} \geq Q_{\text{veto}}
\end{equation}

where $V_{\text{veto}}$ is the number of citizens registering a veto and $Q_{\text{veto}}$ is the veto quorum (typically higher than $Q$ to prevent minority reversal of majority decisions). A successful veto returns the proposal to state \textit{Vetoed}, triggering re-deliberation.

\section{Legitimacy Score}

We define a composite legitimacy score for a decision $d$:

\begin{equation}
  \label{eq:legitimacy}
  L(d) = w_1 \cdot R(d) + w_2 \cdot \frac{V_{\text{for}}(d)}{V(d)} + w_3 \cdot \Delta_{\text{delib}}(d)
\end{equation}

where $w_1 + w_2 + w_3 = 1$, and $\Delta_{\text{delib}}(d) \in [0,1]$ is a structured deliberation quality score decomposed into four measurable components:

\begin{equation}
  \label{eq:delib}
  \Delta_{\text{delib}}(d) = a_1 \cdot A(d) + a_2 \cdot E(d) + a_3 \cdot D(d) + a_4 \cdot I(d), \qquad a_1 + a_2 + a_3 + a_4 = 1
\end{equation}

where:
\begin{itemize}
  \item $A(d) \in [0,1]$: \textbf{argument diversity} --- the number of distinct argument positions recorded in the deliberation archive, normalised by the maximum observed across comparable decision categories.
  \item $E(d) \in [0,1]$: \textbf{expert coverage} --- the fraction of identified impact domains that received formal expert assessment during deliberation.
  \item $D(d) \in [0,1]$: \textbf{deliberation depth} --- the number of proposal revision cycles, normalised by the category maximum, indicating iterative refinement rather than initial adoption.
  \item $I(d) \in [0,1]$: \textbf{inclusive participation} --- the fraction of recognised stakeholder groups that contributed at least one argument.
\end{itemize}

All four components are computable directly from the deliberation archive without subjective quality judgements. The weights $a_1$--$a_4$ and the outer weights $w_1$--$w_3$ require empirical calibration against participant outcome data; their estimation is identified as a first-priority research output of Tier~1 deployments (Chapter~\ref{sec:pathway}). This composite score enables comparative evaluation of decisions across governance instances and longitudinal tracking of institutional deliberation health.

% ══════════════════════════════════════════════════════════════════
\chapter{Identity Model}
\label{sec:identity}

The correctness of the participation rate $R = V/P$ depends critically on $P$ being accurate---that is, that each eligible person is counted exactly once. This is the identity problem.

\section{Threat Model}

Digital governance faces three primary identity threats:

\begin{enumerate}
  \item \textbf{Sybil Attacks.} A single actor creates multiple identities, inflating their voting weight \citep{douceur2002sybil}.
  \item \textbf{Identity Duplication.} Legitimate identities are duplicated across jurisdictions or systems.
  \item \textbf{Coercion.} Eligible citizens are coerced into voting in particular ways or selling their votes.
\end{enumerate}

\section{Identity Model Comparison}

Table~\ref{tab:identity} compares four identity approaches across the four properties most relevant to governance-scale deployment.

\begin{table}[htbp]
\centering
\caption{Comparison of identity models for digital governance.}
\begin{tabular}{@{} p{3.5cm} p{2.2cm} p{2.2cm} p{2.2cm} p{2.2cm} @{}}
\toprule
\textbf{Model} & \textbf{Sybil Resistance} & \textbf{Privacy} & \textbf{Scalability} & \textbf{Uniqueness} \\
\midrule
Centralised (state ID)     & High & Low  & High & High \\
Decentralised (self-sovereign) & Medium & High & Medium & Low \\
Proof of Personhood (biometric) & Very High & Very Low & High & Very High \\
Social Graph Verification  & Medium--High & Medium & Medium & Medium \\
Hybrid (proposed)          & High & Medium--High & High & High \\
\bottomrule
\end{tabular}
\label{tab:identity}
\end{table}

\section{Proposed Hybrid Model}

We propose a hybrid identity approach combining:

\begin{itemize}
  \item \textbf{Social-graph verification.} Graph-based algorithms identify clusters of mutually-vouching identities, making large-scale Sybil attacks computationally expensive \citep{pentland2014social}.
  \item \textbf{Cryptographic proofs.} Zero-knowledge proofs allow citizens to prove eligibility (e.g., residency, age) without revealing identifying information \citep{boneh2020graduate}.
  \item \textbf{Periodic revalidation.} Identities are revalidated on a rolling basis, preventing accumulation of dormant fake accounts.
\end{itemize}

The privacy-accountability trade-off is unavoidable: stronger identity verification generally requires more personal data. PPG positions this trade-off explicitly in the governance configuration, allowing different jurisdictions to make different choices within defined bounds.

A specific limitation of the SBT-based model deserves noting here. SBT non-transferability holds only as long as the credential holder maintains exclusive control of their private key: a compromised or delegated key effectively transfers the governance credential, undermining both the Sybil resistance and coercion-resistance properties the model is designed to provide. Key management failures are the primary security failure mode in deployed cryptographic systems at population scale. The full scope of this limitation and the mitigation strategies available at each deployment tier are addressed in Section~\ref{sec:limitations}.

% ══════════════════════════════════════════════════════════════════
\chapter{System Architecture}
\label{sec:architecture}

The PPG system is structured as three logical layers. Each layer has a defined boundary, a single responsibility, and explicit interfaces to adjacent layers. The layered design enforces the public/private separation that is the architectural expression of the Transparency Default (Section~\ref{sec:inversion}): public governance data flows upward through all layers; private citizen identity data never crosses the boundary between the Identity Layer and the Governance Engine.

\section{Identity Layer}

The Identity Layer is responsible for one function: determining whether each presenting citizen is eligible to participate in a given governance decision, and communicating that determination as a binary signal without exposing any personal identifier.

Internally, the Identity Layer maintains three components. A \textbf{hashed eligibility registry} stores externally attested credentials---civic registration, residency attestations, or social-graph verification outputs---as one-way cryptographic commitments; no raw credential data is retained beyond the commitment hash, so a registry breach exposes no personal information. A \textbf{social-graph analysis component} processes the mutual-attestation network to compute Sybil resistance scores and flag potentially duplicated identities (Section~\ref{sec:identity}). A \textbf{ZK proof verification service} accepts zero-knowledge eligibility proofs submitted by citizens, validates them against the commitment registry, and returns a yes/no eligibility signal without recovering the underlying identity.

The Identity Layer exposes two values to the Governance Engine for each active decision: the total eligible population count $P$ for quorum evaluation, and a per-decision participation counter that increments by one for each validated proof. To prevent any citizen from casting more than one vote per decision, each ZK proof incorporates a \emph{participation nullifier}---a deterministic commitment derived from the citizen's credential and the specific decision identifier. The nullifier is published alongside the proof; any subsequent submission producing the same nullifier is rejected by the Governance Engine. Following the Semaphore nullifier construction \citep{boneh2020graduate}, the commitment reveals nothing about the underlying credential, preserving participation anonymity while guaranteeing vote uniqueness. No individual identifier, participation history, or vote direction is communicated across the Identity Layer boundary. The eligible count $P$ is recomputed each governance cycle following a periodic revalidation sweep that removes expired, flagged, or unrenewed credentials.

\section{Governance Engine}

The Governance Engine implements the formal state machine of Chapter~\ref{sec:formal}. Its responsibilities span the full decision lifecycle.

\textbf{Proposal management.} On receiving a new submission from the Interface Layer, the Engine validates that all required fields are present (problem statement, proposed action, impact scope, affected stakeholders, resource requirements) and creates a proposal record in state \textit{Proposed}. The proposal identifier, submission timestamp, and submitter eligibility proof hash (not identity) are appended to the public audit log.

\textbf{Deliberation coordination.} The Engine opens the deliberation period and manages its timer. It receives structured argument contributions from the Interface Layer, validates their format, and appends them to the public deliberation archive. At deliberation close, it computes the argument diversity and deliberation depth components of the quality score $\Delta_{\text{delib}}$ and triggers constitutional validation.

\textbf{Quorum enforcement.} When a vote closes, the Engine evaluates Equations~\ref{eq:pass}--\ref{eq:dynamic_quorum} atomically: it reads the eligible count $P$ from the Identity Layer, counts validated vote signals, and applies the dynamic quorum threshold $Q(i)$. State transitions occur only when the threshold condition is satisfied; no state change can be made by any operator outside this enforcement logic.

\textbf{Veto registry.} Following execution, the Engine opens the veto window of duration $\tau$ and accepts veto registrations from the Interface Layer. When the window closes, it evaluates Equation~\ref{eq:veto} and transitions to \textit{Vetoed} or confirms \textit{Executed} (terminal) accordingly.

\textbf{Audit log.} Every state transition---including the triggering event, timestamp, and relevant counts---is appended to a cryptographically signed public log. The log is append-only; no entry can be modified or deleted retroactively. The complete lifecycle history of every proposal is publicly queryable at any time.

\section{Interface Layer}

The Interface Layer translates citizen intent into structured inputs for the Governance Engine and presents governance state in accessible form. It holds no persistent governance state of its own.

\textbf{Submission interface.} A structured form guides citizens through the required proposal fields with contextual help and validation. Completed submissions are forwarded to the Governance Engine; the Interface Layer retains no copy.

\textbf{Deliberation interface.} During the deliberation period, the Interface Layer presents the current argument map and allows citizens to contribute arguments tagged by type (factual claim, normative claim, procedural objection, proposed amendment). Expert assessments are surfaced alongside citizen arguments. A public archive, queryable by any citizen, presents the full argument record chronologically and by type, and feeds the automated argument diversity scoring.

\textbf{Voting interface.} At vote open, the Interface Layer locally constructs the citizen's ZK eligibility proof from their credential wallet, submits it to the Identity Layer for verification, and---on receiving a positive eligibility signal---presents the ballot. The vote direction is submitted directly to the Governance Engine; the Interface Layer records only that a proof was submitted, not its outcome or the vote direction.

\textbf{Veto and notification interfaces.} During the veto window, citizens register a veto through the same ZK proof flow. The notification service alerts eligible citizens to newly opened deliberation, voting, and veto windows within their scope.

\textbf{Accessibility requirements.} The Interface Layer is specified to WCAG~2.1~AA standards, with low-bandwidth modes, screen reader compatibility, multi-language support, and a low-literacy interaction path. Physical participation equivalents---accessible terminals at civic facilities, in-person assistance programmes---connect to the same Interface Layer API, ensuring non-digital participation is recorded in the same public audit log as digital participation.

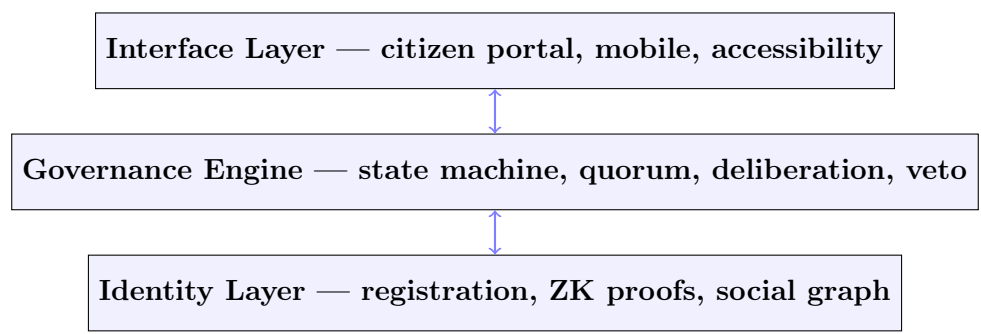
\begin{figure}[htbp]
\centering
\begin{tikzpicture}[node distance=1.4cm,
  layer/.style={draw, rectangle, minimum width=9cm, minimum height=1cm,
                fill=blue!6, font=\small\bfseries, text centered}]
  \node (ui)  [layer] {Interface Layer --- citizen portal, mobile, accessibility};
  \node (eng) [layer, below of=ui, node distance=1.6cm]
              {Governance Engine --- state machine, quorum, deliberation, veto};
  \node (id)  [layer, below of=eng, node distance=1.6cm]
              {Identity Layer --- registration, ZK proofs, social graph};

  \draw[<->, thick, blue!50] (ui) -- (eng);
  \draw[<->, thick, blue!50] (eng) -- (id);
\end{tikzpicture}
\caption{Three-layer PPG system architecture.}
\label{fig:architecture}
\end{figure}

\section{Public/Private Separation by Design}
\label{sec:pubpriv}

The Transparency Default (Section~\ref{sec:inversion}) is operationalised in the three-layer architecture through specific data-flow constraints at each layer. The Identity Layer stores minimum necessary citizen data---hashed eligibility attestations---and exposes only aggregate eligible counts to the Governance Engine; no individual voting record, participation history, or personal identifier crosses this boundary. The Governance Engine records all proposal states, deliberation contributions, vote tallies, execution logs, and veto registrations as public append-only data; every state transition is timestamped, cryptographically signed by the executing contract, and publicly queryable. There is no private state in the Governance Engine. The Interface Layer authenticates citizen eligibility via ZK proof---receiving only a yes/no eligibility signal from the Identity Layer---and records no session data beyond what is needed to provide the interface service.

This design pattern follows Privacy by Design principles \citep{boneh2020graduate}: privacy protection is built into the system architecture rather than bolted on as a policy layer. A technical guarantee is stronger than a legal guarantee: a policy prohibition on surveillance can be reversed by a subsequent decision; an architectural impossibility cannot. This distinction is significant for governance systems that may operate across long timescales and through changes in political leadership.

% ══════════════════════════════════════════════════════════════════
\chapter{Simulation Methodology and Results}
\label{sec:simulation}

To validate the framework's dynamic properties, we constructed a stochastic simulation with three behavioural archetypes operating over 100 decision rounds.

\begin{mdframed}
\textbf{Modelling scope.} This is a stylised stochastic model, not a full agent-based model: each archetype is characterised by a fixed base participation rate and a random perturbation around it. Agents do not interact with one another, hold no memory of prior rounds, and do not update strategies in response to outcomes. The model is designed to validate the existence of parameter regimes in which the framework's dynamic properties hold (stable approval convergence, functional veto dynamics, manipulation deterrence), not to replicate the adaptive complexity of real populations. A full agent-based model with social-network structure, issue-salience dynamics, and strategic learning is an objective of future empirical work.
\end{mdframed}

\section{Simulation Parameters}

Table~\ref{tab:simparams} specifies all quantitative parameters used in the simulation experiments. The qualitative results---participation convergence, veto stability, and manipulation deterrence---are robust to moderate perturbation of these values; the specific settings represent a plausible baseline deployment scenario at municipal scale.

\begin{table}[htbp]
\centering
\caption{Simulation parameters.}
\label{tab:simparams}
\begin{tabular}{@{} l l p{6.5cm} @{}}
\toprule
\textbf{Parameter} & \textbf{Value} & \textbf{Description} \\
\midrule
$T$ & 100 & Number of simulated decision rounds \\
$Q_{\text{base}}$ & $\{0.2,\,0.3,\,0.4\}$ & Baseline quorum threshold values tested \\
Passive rate & 0.14 & Base participation rate, passive archetype \\
Active rate  & 0.34 & Base participation rate, active archetype \\
Strategic rate & 0.52 & Base participation rate, strategic archetype \\
$g_{\max}$ & 0.80 & Maximum expected utility gain from manipulation \\
$k$         & 6.0  & Gain function steepness parameter \\
$c_1$       & 0.55 & Quadratic mobilisation cost coefficient \\
$c_2$       & 0.08 & Linear mobilisation cost coefficient \\
\bottomrule
\end{tabular}
\end{table}

\section{Agent Archetypes}

\begin{description}
  \item[Passive Agents] Participate rarely and only when issues are highly salient. Base participation rate $\approx 0.14$, increasing slowly over time.
  \item[Active Agents] Regularly engaged citizens with consistent participation. Base participation rate $\approx 0.34$.
  \item[Strategic Agents] Self-interested actors who coordinate participation to influence outcomes. Base participation rate $\approx 0.52$.
\end{description}

\begin{figure}[htbp]
\centering
\includegraphics[width=0.85\textwidth]{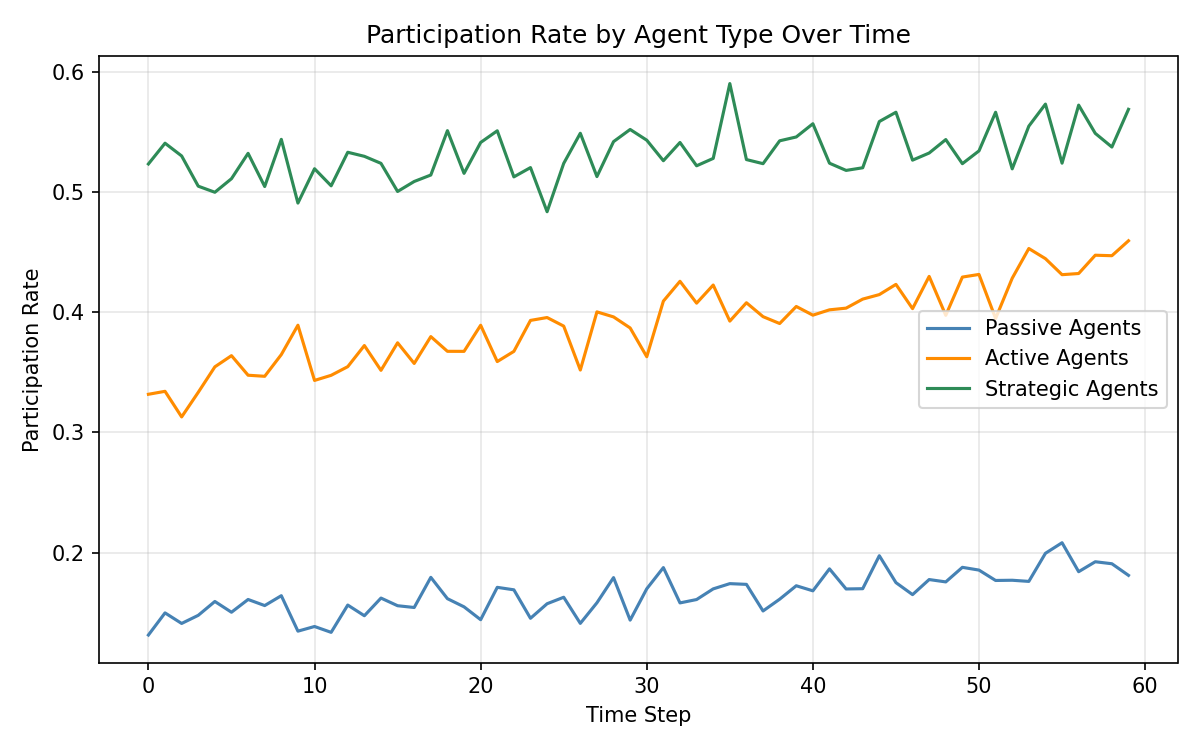}
\caption{Participation rates by agent type over 60 simulated decision rounds. Strategic agents consistently achieve the highest participation; passive agents' base participation rate increases gradually over time.}
\label{fig:agents}
\end{figure}

\section{Quorum Sensitivity}

\begin{figure}[htbp]
\centering
\begin{subfigure}[b]{0.48\textwidth}
  \includegraphics[width=\textwidth]{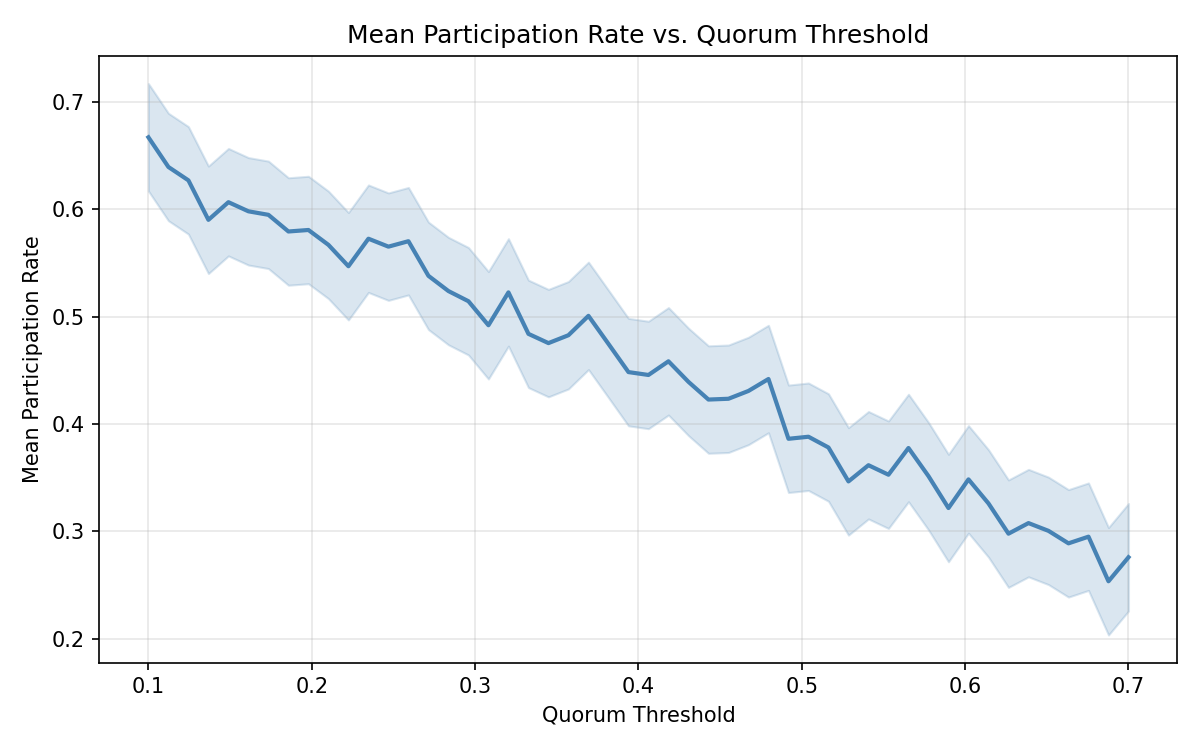}
  \caption{Mean participation vs.\ quorum threshold.}
\end{subfigure}
\hfill
\begin{subfigure}[b]{0.48\textwidth}
  \includegraphics[width=\textwidth]{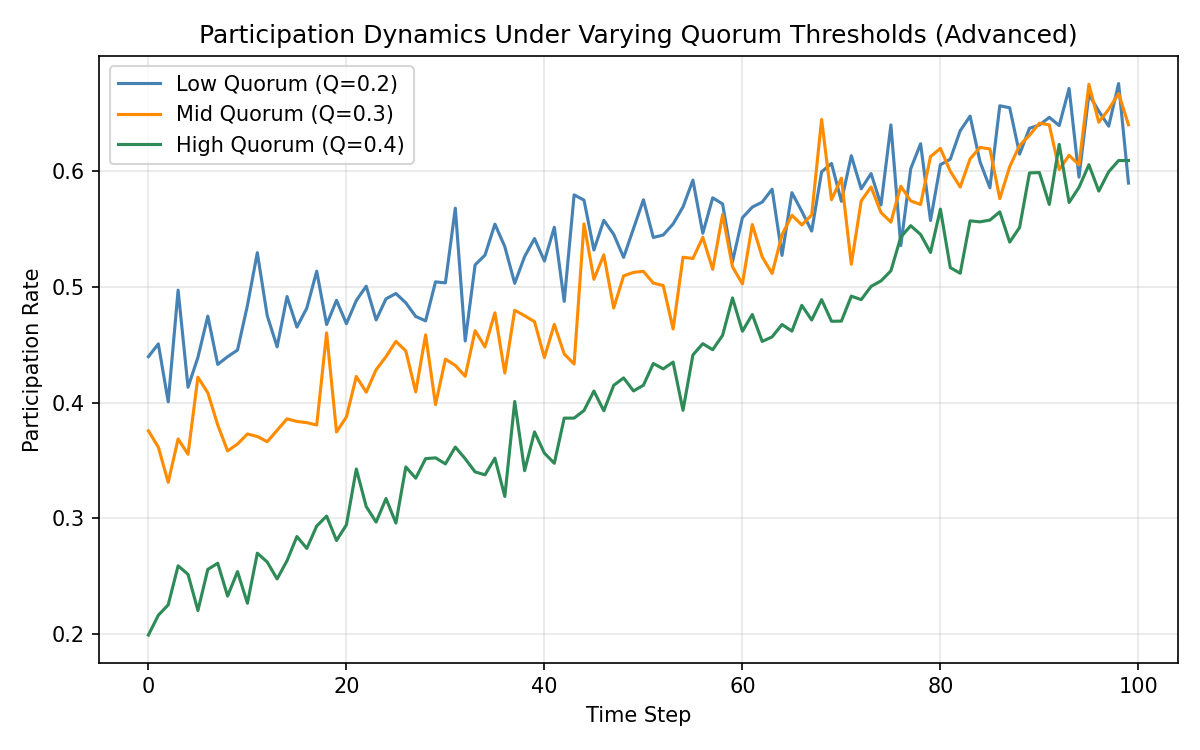}
  \caption{Participation dynamics under varying quorum thresholds.}
\end{subfigure}
\caption{Quorum sensitivity analysis. Higher thresholds increase mean participation but reduce decision throughput, illustrating the trade-off the dynamic quorum formula mediates.}
\label{fig:quorum_sensitivity}
\end{figure}

As quorum thresholds increase, mean participation rises (Figure~\ref{fig:quorum_sensitivity}a), but decision throughput falls. The dynamic quorum model (Equation~\ref{eq:dynamic_quorum}) calibrates this trade-off per decision, reserving high quorum requirements for high-impact proposals.

\section{System Stability and Veto Dynamics}

\begin{figure}[htbp]
\centering
\begin{subfigure}[b]{0.48\textwidth}
  \includegraphics[width=\textwidth]{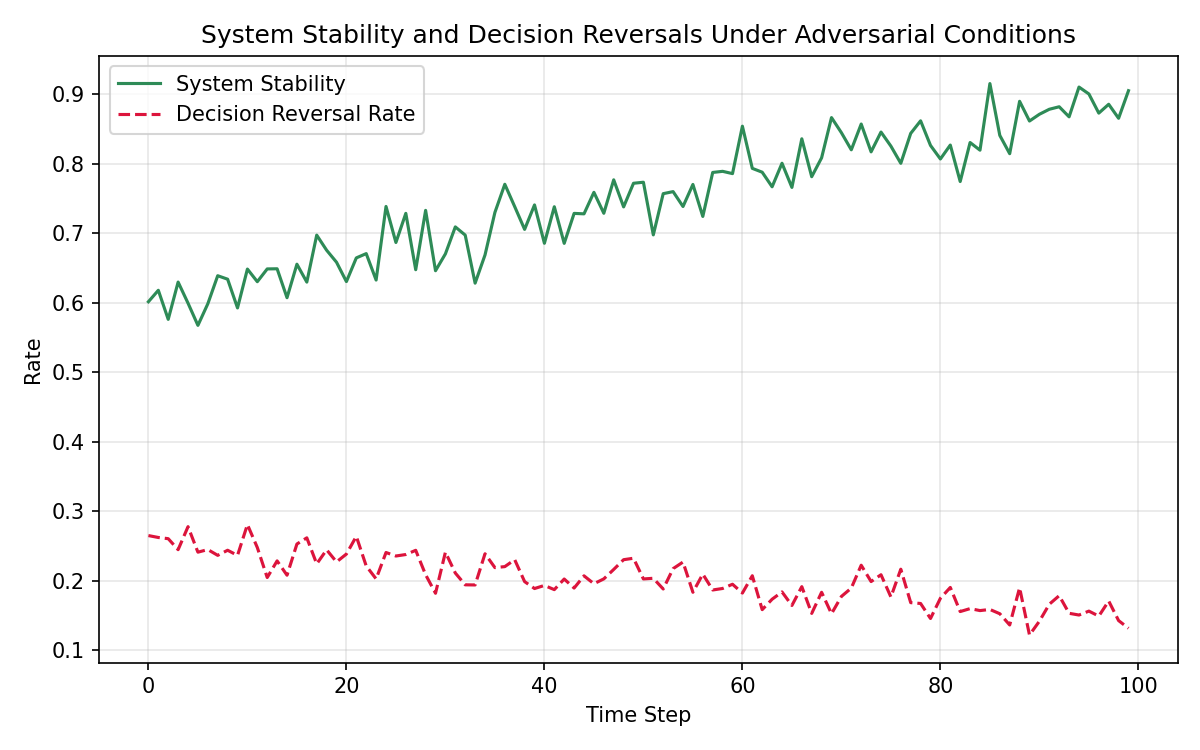}
  \caption{System stability and reversal rate.}
\end{subfigure}
\hfill
\begin{subfigure}[b]{0.48\textwidth}
  \includegraphics[width=\textwidth]{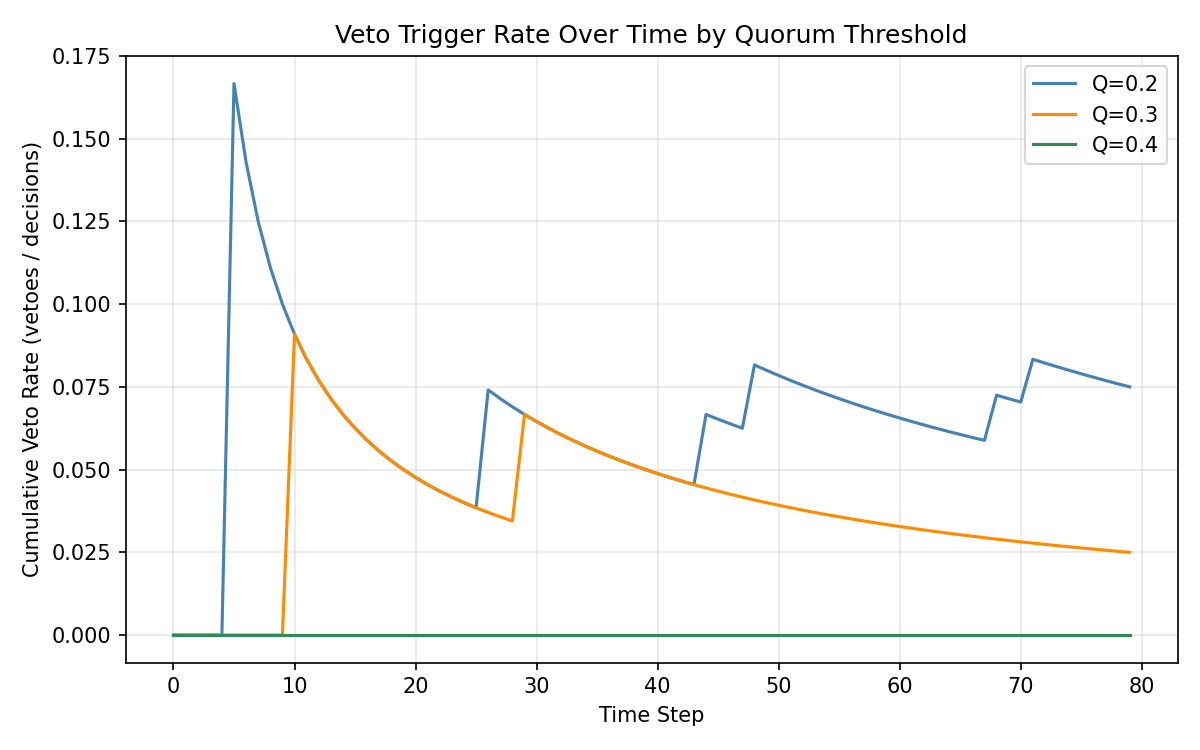}
  \caption{Cumulative veto trigger rate by quorum threshold.}
\end{subfigure}
\caption{System stability under adversarial conditions. Stability increases over time even as reversal rates decrease, confirming the veto mechanism's role as a safety valve rather than a destabiliser.}
\label{fig:stability}
\end{figure}

Figure~\ref{fig:stability}a shows that system stability increases over time even as decision reversal rates decrease, indicating that the veto mechanism functions as a self-correcting safety valve rather than a destabilising feature. Figure~\ref{fig:stability}b demonstrates that higher quorum thresholds suppress veto rates by increasing the threshold for both passage and reversal.

\section{Approval Rate}

\begin{figure}[htbp]
\centering
\includegraphics[width=0.8\textwidth]{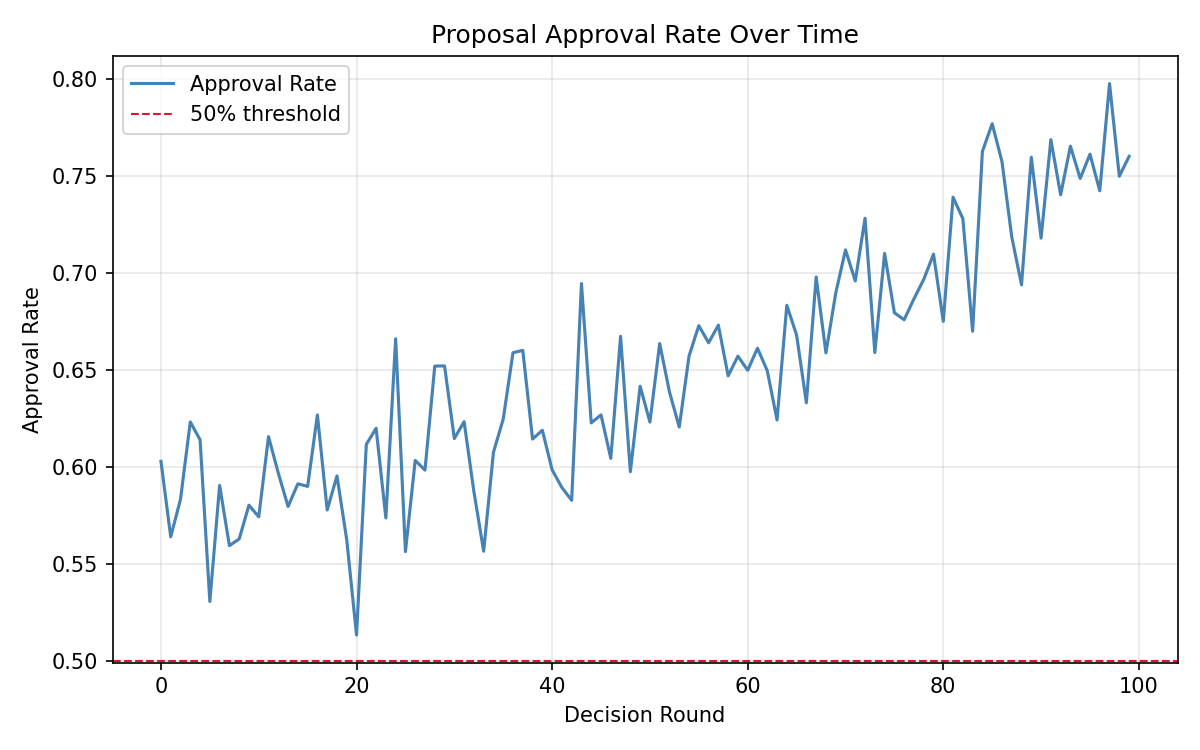}
\caption{Proposal approval rate over 100 decision rounds. Rates converge above the 50\% threshold, reflecting the emergence of stable participatory norms across agent archetypes.}
\label{fig:approval}
\end{figure}

Approval rates stabilise above the decision threshold over time, reflecting the emergence of participatory norms in the agent population.

% ══════════════════════════════════════════════════════════════════
\chapter{Game-Theoretic Analysis}
\label{sec:gametheory}

A critical question for any governance system is whether it is manipulation-resistant: can a coordinated minority extract disproportionate benefits by strategic participation?

\section{Manipulation Model}

Let $f \in [0, 1]$ denote the fraction of the electorate that a coordinated faction controls. The expected utility gain from manipulation is:

\begin{equation}
  \label{eq:gain}
  G(f, Q) = g_{\max} \left(1 - e^{-k \cdot \max(f - Q, 0)}\right)
\end{equation}

where $g_{\max}$ is the maximum utility gain, $k$ controls the steepness of the gain function, and $\max(f - Q, 0)$ captures the margin by which the faction exceeds the quorum threshold.

The mobilisation cost of controlling fraction $f$ of the electorate is:

\begin{equation}
  \label{eq:cost}
  C(f) = c_1 f^2 + c_2 f
\end{equation}

reflecting the quadratic cost of mobilising a large fraction of the population.

In the simulation illustrated in Figure~\ref{fig:gametheory}, the parameters take values $g_{\max} = 0.80$, $k = 6.0$, $c_1 = 0.55$, and $c_2 = 0.08$ (see also Table~\ref{tab:simparams}). These represent a scenario in which the maximum achievable manipulation gain is bounded (normalised to $[0,1]$), mobilisation costs are convex-dominant, and the linear term $c_2 f$ ensures strictly positive marginal cost at zero faction scale. The qualitative conclusions of Theorem~\ref{thm:manipulation}---the existence of $f^{*} > 0$ and its monotone dependence on $Q$---hold for any strictly positive $g_{\max}, k, c_1, c_2$.

\section{Equilibrium}

\begin{theorem}[Manipulation Deterrence]
\label{thm:manipulation}
For any quorum $Q > 0$, there exists a critical faction size $f^* > 0$ such that for all $f < f^*$, the manipulation cost exceeds the expected gain: $C(f) > G(f, Q)$. Coordination at scale $f < f^*$ is therefore not a Nash equilibrium.
\end{theorem}

\begin{proof}[Proof sketch]
Set $f^* = \sup\{f : C(f) \geq G(f, Q)\}$. At $f = 0$, both functions are zero. The cost function has $C'(0) = c_2 > 0$. The gain function satisfies $G'(0) = 0$, because $\max(f - Q, 0) = 0$ for all $f \leq Q$, so the gain is identically zero until the faction exceeds the quorum threshold---the gain curve is flat in $[0, Q]$. Since the cost curve rises strictly from the origin while the gain curve remains flat, $C(f) > G(f, Q)$ for all sufficiently small $f > 0$. By continuity of both functions, $f^*$ is well-defined and positive. For any $f < f^*$, joining the coordinating faction incurs net cost; rational agents therefore decline to join, and the faction cannot reach scale $f^*$. \qed
\end{proof}

\begin{figure}[htbp]
\centering
\includegraphics[width=0.85\textwidth]{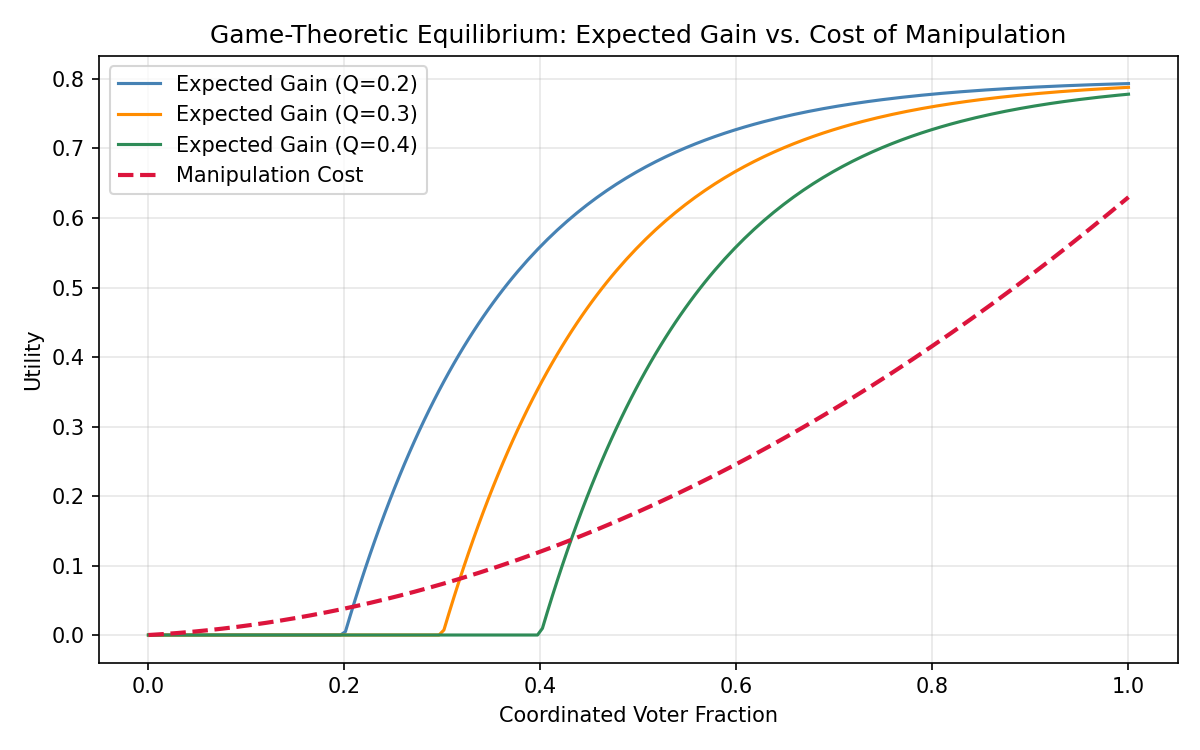}
\caption{Expected utility gain versus manipulation cost for varying quorum thresholds. Higher quorums shift the equilibrium point rightward, requiring a larger faction for manipulation to be profitable.}
\label{fig:gametheory}
\end{figure}

Figure~\ref{fig:gametheory} illustrates that higher quorum thresholds raise the manipulation cost crossover point, making coordinated manipulation harder. This is one of the key arguments for dynamic rather than fixed quorum thresholds.

\section{Repeated Games and Collusion Dynamics}

The single-round result in Theorem~\ref{thm:manipulation} establishes a safe quorum threshold but leaves open a more subtle threat: a faction too small to profitably manipulate any single decision might benefit from a long-run strategy of incremental capture---participating selectively across many rounds to shift governance outcomes gradually.

We model the governance system as an infinitely repeated game in which the same set of players participates in successive rounds with discount factor $\beta \in (0,1)$. Each faction $j$ chooses a participation strategy in each round $t$; the per-round payoff structure is as in the single-round model and the long-run payoff is the discounted sum $\sum_{t=0}^{\infty} \beta^t \pi_j^t$.

\begin{proposition}[Collusion Containment Under Dynamic Quorum]
\label{prop:collusion}
In the repeated governance game with dynamic quorum $Q(i) = Q_{\text{base}} + \alpha \cdot \sigma(i)$, sustaining a collusion equilibrium in which a faction of size $f < f^*$ achieves positive long-run gain requires discount factor $\beta \geq \beta^*(Q_{\text{base}}, \alpha)$, where $\beta^*$ is strictly increasing in both $Q_{\text{base}}$ and $\alpha$. Sufficiently patient factions can sustain collusion; sufficiently impatient factions cannot.
\end{proposition}

\begin{proof}[Proof sketch]
In each round the collusion faction's net gain is $G(f, Q(i)) - C(f)$. For $f < f^*$, single-round gain is negative, so the faction's strategy is to absorb per-round costs in exchange for long-run capture payoffs once its position is established. The critical discount factor $\beta^*$ solves the indifference condition between defecting (taking the single-round outside payoff) and sustaining the collusion agreement (discounted long-run stream). The monotonicity of $\beta^*$ in $Q_{\text{base}}$ and $\alpha$ follows from the fact that both parameters enter $Q(i)$ additively and $G(f,Q)$ is strictly decreasing in $Q$: a higher quorum reduces the per-round gain available to the colluding faction, shrinking the long-run surplus and requiring greater patience for the strategy to dominate. Formally, differentiating the indifference condition with respect to $Q_{\text{base}}$ and $\alpha$ and applying the implicit function theorem yields $\partial \beta^*/\partial Q_{\text{base}} > 0$ and $\partial \beta^*/\partial \alpha > 0$, given the strict monotonicity of $G$ in $Q$. Existence of a finite $\beta^*$ for any positive $Q_{\text{base}}$, $\alpha$, and $f < f^*$ follows from the folk theorem \citep{fudenberg1986folk}.
\end{proof}

This is the governance analogue of the folk theorem for infinitely repeated games \citep{fudenberg1986folk}: long-run collusion is sustainable for sufficiently patient actors but becomes increasingly costly as quorum thresholds rise. The empirical record of cooperation sustaining in repeated-interaction settings among self-interested agents \citep{axelrod1984evolution} underscores the practical relevance of this threat; Proposition~\ref{prop:collusion} characterises the parameter conditions under which the design resists it. The strategic implications for PPG parameter design are twofold. First, the impact sensitivity parameter $\alpha$ deters collusion as well as single-round manipulation: higher $\alpha$ raises $\beta^*$, requiring colluders to be more patient---and patient colluders are detectable, since their strategy must persist across many observable rounds. Second, the veto mechanism introduces a cross-round accountability structure absent from the single-round model: a faction that captures a high-impact decision faces increased veto mobilisation in subsequent rounds from citizens whose preferences were overridden, making sustained capture progressively costlier. The formal analysis of this veto-as-reputational-deterrent dynamic requires empirical discount rate estimates and is deferred to post-deployment work.

% ══════════════════════════════════════════════════════════════════
\chapter{System Properties}
\label{sec:properties}

\section{Legitimacy}

Legitimacy in PPG is \emph{measurable} rather than assumed. The legitimacy score $L(d)$ (Equation~\ref{eq:legitimacy}) provides a quantitative basis for comparing decisions and for tracking institutional health over time. Decisions consistently scoring below a threshold signal governance failure.

\section{Transparency}

All governance actions---proposal submission, deliberation contributions, votes, execution records---are logged immutably. This auditability satisfies Habermas's communicative rationality criterion: any citizen can inspect the reasoning behind any decision \citep{habermas1996between}. Transparency here is not a feature but a structural default, implementing the Transparency Default (Section~\ref{sec:inversion}) at the data layer.

\section{Accountability}

Accountability is structural rather than electoral. The veto mechanism provides real-time accountability, while the legitimacy score creates an ongoing public record of institutional performance. This corresponds to Rosanvallon's counter-democratic power of oversight \citep{rosanvallon2008counter}: accountability flows continuously from government to citizens rather than episodically at election intervals.

\section{Iterative Governance}

PPG is explicitly designed as an \emph{iterative system} rather than a finished solution. Governance operates in a continuously evolving social, technological, and institutional context; fixed governance rules are fragile and become brittle over time. Ostrom's polycentric governance theory directly anticipates this: effective commons institutions are not designed once and frozen but evolve through practice, observation, and collective choice \citep{ostrom1990governing}. \citet{hayek1945use} reaches a structurally identical conclusion from a different tradition: because relevant knowledge is inherently dispersed and tacit, no fixed institutional design can claim to have correctly specified all governance parameters in advance; the capacity for adaptive rule modification is not an optional refinement but an epistemic requirement.

This convergence between institutional theory and epistemology has a formal correlate in mechanism design. Hurwicz's foundational result establishes that a centralised planner cannot correctly elicit agents' private valuations without satisfying incentive compatibility constraints that require active participation \citep{hurwicz1973design}; Maskin's extension characterises the full class of social goals implementable under these constraints \citep{maskin2008mechanism}. Applied to governance, these results provide a formal basis for the epistemic argument:

\begin{proposition}[Epistemic Necessity of Participatory Input]
\label{prop:epistemic}
Let $\theta_i$ denote the private governance preferences held by citizen $i$ (local knowledge, valuation of outcomes). A centralised governance mechanism that does not elicit $\theta_i$ from the eligible population will systematically produce decisions based on incomplete preference information, making them liable to be Pareto-inferior to those achievable with full participation. A participatory mechanism that aggregates revealed preferences through structured deliberation and verifiable voting is epistemically superior to centralised administrative discretion: not merely normatively preferable, but informationally necessary for governance that reliably reflects the population's actual valuations.
\end{proposition}

\begin{proof}[Proof sketch]
Hurwicz establishes that any social choice function not satisfying incentive compatibility induces rational agents to misreport private information \citep{hurwicz1973design}. In the governance setting, $\theta_i$ encodes citizen $i$'s local knowledge, preference intensities, and the expected consequences of a proposal---information that is private, dispersed, and irreducible to any publicly observable signal. A centralised mechanism that fixes decisions without eliciting $\theta_i$ implements a social choice function with no revelation mechanism. By Hurwicz's result, rational agents misreport when not given an incentive-compatible elicitation mechanism, so the decisions taken reflect distorted signals rather than true preferences $\theta_i$. By Maskin's characterisation \citep{maskin2008mechanism}, the set of social goals implementable under incentive compatibility strictly subsumes those achievable by a non-participatory mechanism. A non-participatory mechanism therefore makes decisions on systematically corrupted preference information; any outcome that would differ from the truthfully-revealed optimum represents an avoidable Pareto-inferior choice, except when all $\theta_i$ happen to be uniform across the population. The PPG pipeline provides the structural conditions for preference revelation through structured argument contribution (deliberation layer) and verifiable vote direction (voting layer), approximating the incentive-compatible elicitation conditions under the empirically supported assumption that civic participation incentives are sufficient to motivate truthful expression. \qed
\end{proof}

Proposition~\ref{prop:epistemic} has a direct implication for the dynamic quorum design: if participation is not merely a legitimacy requirement but an epistemic input, then decisions reached on low participation are not only less legitimate---they are less likely to be correct. The dynamic quorum formula (Equation~\ref{eq:dynamic_quorum}) is therefore an epistemic filter as well as a normative instrument, calibrating the required participation to the stakes of the decision.

This self-referential property is detailed in Section~\ref{sec:conversion}: the framework governs its own parameters through the same pipeline it describes, and the staged deployment model (Chapter~\ref{sec:pathway}) ensures that each tier generates empirical feedback that refines the protocol before scaling.

\section{Resilience}

Simulation evidence (Chapter~\ref{sec:simulation}) demonstrates stable approval convergence under adversarial strategic participation across all tested quorum settings. The veto mechanism functions as a self-correcting safety valve rather than a destabilising reversal mechanism---Figure~\ref{fig:stability} shows that system stability increases over time even as reversal rates decrease. The analytical bound on adversarial faction size is established by Theorem~\ref{thm:manipulation}: no faction below the critical size $f^*$ can expect positive net gains from coordinated manipulation under the model's parameters.

\section{Scalability}

Decision pipelines are parallelisable: independent proposals proceed concurrently through all stages, so the throughput of the governance system scales with the number of concurrent proposals rather than being bottlenecked by sequential processing. The dynamic quorum formula partitions decisions by required mobilisation without central coordination, allowing low-stakes decisions to clear quickly while high-impact decisions draw broader participation. The identity verification architecture (Chapter~\ref{sec:identity}) is explicitly designed for national-scale populations, using zero-knowledge proofs and social-graph attestation to verify eligibility without storing personal data centrally.

% ══════════════════════════════════════════════════════════════════
\chapter{Towards a Decentralised Implementation}
\label{sec:decentralised}

The PPG framework developed in previous sections is implementation-agnostic: it can be deployed on centralised civic infrastructure, federated institutional networks, or fully decentralised programmable systems. We argue here that a decentralised implementation is not merely technically feasible but is, in important respects, the \emph{most suitable} substrate for PPG---and that this paper constitutes the theoretical and formal groundwork toward that goal. The five-layer architecture specified here decomposes the Governance Engine of Chapter~\ref{sec:architecture} into Execution and Coordination layers and adds a Settlement Layer below; the Identity and Interface layers map directly to their Chapter~\ref{sec:architecture} counterparts.

\section{The Case for Decentralisation}

Centralised governance infrastructure introduces three structural vulnerabilities that directly undermine PPG's design objectives:

\begin{enumerate}
  \item \textbf{State capture.} A centralised governance platform controlled by any single operator creates a single point of capture. If the operator modifies quorum rules, suppresses proposals, or alters vote tallies, the formal guarantees of PPG are void. Immutable, publicly auditable execution eliminates this threat surface.
  \item \textbf{Censorship and suppression.} Centralised systems can be compelled---legally or politically---to remove proposals, silence participants, or manipulate deliberation feeds. Censorship-resistant infrastructure ensures the governance process cannot be shut down by any party, including the state.
  \item \textbf{Trust bootstrapping.} Citizens must trust the operator of a centralised system. This recurring obstacle is particularly acute in low-trust institutional environments. A decentralised system replaces institutional trust with cryptographic verification: \emph{don't trust, verify.}
\end{enumerate}

These concerns motivate the framing of PPG as a \emph{governance protocol}---a specification that can be executed trustlessly on decentralised infrastructure, just as TCP/IP specifies communication without depending on any single network operator \citep{defilippi2018blockchain}.

\section{Lessons from Existing DAO Governance}

Decentralised Autonomous Organisations (DAOs) are the first large-scale experiments in programmable, blockchain-based governance \citep{hassan2021decentralized}. They demonstrate that programmable governance is technically feasible and economically significant: Ethereum-based DAOs collectively manage billions of dollars in on-chain assets. However, their governance models exhibit well-documented structural failures that PPG is designed to correct:

\begin{table}[htbp]
\centering
\caption{Structural analysis of representative DAO governance models.}
\begin{tabular}{@{} l p{2.6cm} p{2.7cm} p{3cm} @{}}
\toprule
\textbf{Protocol} & \textbf{Governance Model} & \textbf{Strengths} & \textbf{Failure Modes} \\
\midrule
Compound / Uniswap & Token-weighted on-chain voting & Transparent, auditable & Plutocracy; participation $<$10\% \\
Nouns DAO           & Daily auction + token vote & High engagement, cultural cohesion & Fragmentation under adversarial forks \\
MakerDAO            & Delegated token vote & Expert delegation, stability & Capture by large token holders \\
Snapshot            & Off-chain signalling, token weight & Low cost, high participation & Non-binding; no execution guarantees \\
Aragon              & Modular governance plugin stack & Composable, extensible & Plugin complexity; dependency risk \\
\bottomrule
\end{tabular}
\label{tab:dao}
\end{table}

The common thread across all models is that one-token-one-vote conflates economic stake with governance standing---a form of plutocracy structurally at odds with democratic principles \citep{buterin2022desoc, chohan2017dao}. PPG's one-person-one-vote model, grounded in verified identity rather than asset ownership, directly addresses this deficit. The Snapshot pattern (off-chain deliberation, on-chain commitment) provides a useful architectural precedent for the Coordination Layer described below, but requires the binding execution guarantee that only on-chain quorum enforcement provides.

\section{DeSoc and Soulbound Governance Primitives}

\citet{buterin2022desoc} propose \emph{Decentralised Society} (DeSoc) as a reorientation of Web3 away from transferable financial assets toward non-transferable social relationships encoded as \emph{Soulbound Tokens} (SBTs). SBTs are the natural identity primitive for PPG in a decentralised context:

\begin{itemize}
  \item \textbf{Non-transferability.} Governance credentials cannot be bought, sold, or delegated, resisting the plutocratic dynamics of token governance. A person cannot sell their participation right any more than they can sell their citizenship.
  \item \textbf{Social-graph provenance.} SBTs issued by communities, institutions, and peers encode the social graph directly on-chain, enabling the graph-based Sybil resistance described in Section~\ref{sec:identity} without requiring a centralised identity provider.
  \item \textbf{Composable eligibility.} Different governance contexts---local ward, municipality, sector-specific regulatory body---can issue context-specific SBTs, enabling the multi-jurisdictional eligibility scoping that PPG's proposal layer requires.
\end{itemize}

DeSoc is not a competing vision to PPG but the social infrastructure layer on which PPG runs. Where PPG provides the governance protocol, DeSoc provides the identity fabric.

\section{Decentralised System Architecture}

A decentralised PPG implementation decomposes into five architectural layers:

\begin{figure}[htbp]
\centering
\begin{tikzpicture}[node distance=1.7cm,
  layer/.style={draw, rectangle, minimum width=10.5cm, minimum height=1cm,
                fill=blue!6, font=\small\bfseries, text centered}]
  \node (ui)    [layer] {Interface Layer --- progressive web app, local-first, accessibility-first};
  \node (coord) [layer, below of=ui]
                {Coordination Layer --- off-chain deliberation, IPFS storage, Arweave audit log};
  \node (exec)  [layer, below of=coord]
                {Execution Layer --- L2 smart contracts, quorum enforcement, veto logic};
  \node (id)    [layer, below of=exec]
                {Identity Layer --- SBTs, ZK proofs of personhood, social-graph verification};
  \node (net)   [layer, below of=id]
                {Settlement Layer --- Ethereum mainnet, cryptographic finality, public auditability};
  \draw[<->, thick, blue!50] (ui) -- (coord);
  \draw[<->, thick, blue!50] (coord) -- (exec);
  \draw[<->, thick, blue!50] (exec) -- (id);
  \draw[<->, thick, blue!50] (id) -- (net);
\end{tikzpicture}
\caption{Five-layer decentralised PPG architecture, from citizen interface to cryptographic settlement.}
\label{fig:dec_architecture}
\end{figure}
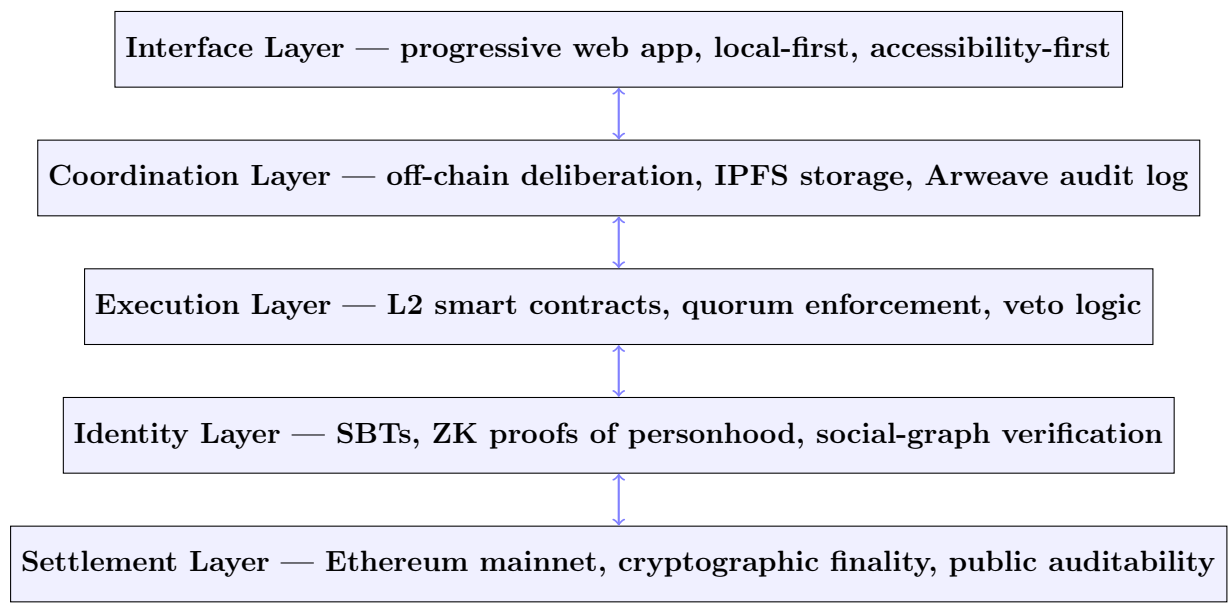

\paragraph{Settlement Layer.} The base layer---Ethereum mainnet or an equivalent programmable settlement network \citep{buterin2014ethereum}---provides cryptographic finality, immutability, and global auditability. Governance \emph{outcomes} (passed proposals, veto records, legitimacy scores) are committed here. The settlement layer is not used for high-frequency operations due to cost and latency constraints.

\paragraph{Identity Layer.} Citizen identities are represented as SBTs and verified through zero-knowledge proofs of eligibility \citep{boneh2020graduate}. No raw personal data is published on-chain; ZK proofs attest to the fact of eligibility without revealing underlying credentials. The social-graph Sybil-resistance algorithm runs off-chain across a distributed validator set, with the resulting eligibility commitment anchored on-chain each governance cycle.

\paragraph{Execution Layer.} The core PPG state machine (Section~\ref{sec:formal}) is implemented as smart contracts on a Layer 2 network---an Optimistic or ZK rollup---for high throughput and low per-transaction costs. The quorum contract enforces Equations~\ref{eq:pass}--\ref{eq:veto} atomically: no state transition occurs without satisfying the on-chain threshold condition. This eliminates the need to trust any intermediary for vote counting.

\paragraph{Coordination Layer.} Deliberation---the most bandwidth-intensive governance activity---is handled off-chain. Proposals, structured arguments, expert assessments, and revision histories are stored on content-addressed networks (IPFS for availability, Arweave for permanent archival). Cryptographic content hashes are anchored on-chain, ensuring the deliberation record cannot be retroactively altered. This extends the Snapshot off-chain signalling pattern with binding execution and tamper-evident archiving.

\paragraph{Interface Layer.} Citizens interact through progressive web applications designed for low-bandwidth environments and broad device compatibility. A local-first architecture ensures the interface remains functional during network partitions, syncing state when connectivity is restored. Multi-language support, screen reader compatibility, and low-literacy interaction modes are built into the specification as first-class requirements.

\section{On-Chain vs.\ Off-Chain Governance Partition}

A central design tension in decentralised governance is the partition between on-chain (fully verifiable, higher cost, slower) and off-chain (fast, low cost, requires trust) operations. We formalise this tension as a partition principle:

\begin{mdframed}
\textbf{Decentralisation Partition Principle:} Place on-chain only what must be \emph{auditable by anyone, enforceable without trust, and permanently recorded}. Place off-chain everything that can be \emph{verified by cryptographic reference} to an on-chain commitment.
\end{mdframed}

Under this principle: vote counts, quorum thresholds, state transitions, and veto records are on-chain; deliberation content, participation history, and identity social graphs are off-chain with on-chain hash commitments. This partition preserves the formal guarantees of the PPG state machine while remaining practically deployable at national scale.

\section{Composability and Governance Modules}

A decentralised PPG stack is \emph{composable}: individual governance modules---quorum contracts, veto registries, identity validators, deliberation anchors---can be assembled into different configurations for different governance contexts without redeploying the full stack. This follows the protocol-layer abstraction articulated by \citet{defilippi2018blockchain}: governance rules are not mere descriptions of desired behaviour but self-executing code---auditable, forkable, and interoperable across jurisdictions.

Composability enables multi-level governance. A neighbourhood assembly, a municipal council, and a sectoral regulatory body can each run PPG-compatible governance contracts while sharing identity infrastructure and audit tooling, interoperating via standardised governance message formats. Governance decisions at lower levels can be escalated to higher tiers through composable proposal-bridging contracts, preserving the legitimacy chain at each level.

\section{Research Roadmap}

The five specific work items this programme requires---smart contract formalisation, ZK identity circuit audit, Sybil resistance testing at scale, Layer~2 performance benchmarking, and pilot deployments---are incorporated into the Tier~1 and Tier~2 research agendas in Chapter~\ref{sec:pathway}, where each item is mapped to the deployment context and empirical prerequisites that make it tractable. The present paper positions itself as \emph{Phase~0} of that programme: providing the formal and empirical foundations on which trustless, censorship-resistant, and democratically legitimate governance infrastructure can be built.

% ══════════════════════════════════════════════════════════════════
\chapter{Limitations}
\label{sec:limitations}

This chapter identifies the significant limitations of the PPG framework, evaluated honestly against the claims made for it. The goal is to provide an accurate picture of what the framework does and does not establish.

\section{Identity: A Partially Solved Problem}

The identity problem---ensuring each eligible person is counted exactly once while preserving privacy---is the most technically challenging component of digital governance. The proposed hybrid model (Soulbound Tokens + social-graph verification + zero-knowledge proofs) is a current best-practice synthesis, but no existing approach fully satisfies all four desired properties simultaneously. Social-graph Sybil resistance degrades as adversaries create organic-looking fake social graphs \citep{douceur2002sybil}; ZK proof circuits are complex and require independent security audits; SBT non-transferability relies on secure key management that many citizens cannot reliably provide. The identity model should be treated as a research specification requiring empirical testing and independent security analysis, not as a production-ready solution.

\section{Simulation Fidelity}

The simulations in Chapter~\ref{sec:simulation} use three stylised behavioural archetypes (passive, active, strategic). Real populations contain a much wider range of participation motivations, shaped by issue salience, social networks, trust levels, and information environments. The simulation model is stochastic rather than fully agent-based: agents do not interact with one another, hold no memory across rounds, and do not adapt strategies in response to outcomes. The simulation results establish the \emph{existence} of parameter regimes in which PPG's dynamic properties hold (stable approval rates, functional veto mechanism, manipulation deterrence), but they do not establish that these regimes match any specific real-world deployment context. Simulation parameters would need to be calibrated against empirical participation data from pilot deployments before the quantitative findings can be applied. The specific base participation rates used in the simulations (passive 14\%, active 34\%, strategic 52\%) are derived from cross-national participation research but are not grounded in any PPG-specific deployment context; Tier~1 pilot data will enable replacement of these estimates with empirically calibrated participation rates and a direct test of whether the dynamic quorum formula maintains predicted throughput behaviour.

The outer legitimacy score weights $w_1, w_2, w_3$ in Equation~\ref{eq:legitimacy} and the inner deliberation quality weights $a_1$--$a_4$ in Equation~\ref{eq:delib} are currently unspecified; their values require empirical calibration against participant outcome data from deployed governance processes. The four-component structure of $\Delta_{\text{delib}}$ (Equation~\ref{eq:delib}) provides an operational measurement specification, but the relative weighting of argument diversity, expert coverage, deliberation depth, and inclusive participation cannot be determined without data on which combinations of these inputs produce decisions that participants subsequently regard as well-considered. The first empirical priority of any Tier~1 deployment should be generating participant satisfaction data sufficient to regress against observed participation and approval rates, providing initial estimates of all seven weights.

\section{Game-Theoretic Model Assumptions}

The manipulation and collusion analysis in Chapter~\ref{sec:gametheory} rests on two functional form choices: a quadratic mobilisation cost $C(f) = c_1 f^2 + c_2 f$ and an exponential gain function $G(f,Q) = g_{\max}(1 - e^{-k\max(f-Q,0)})$. These are standard choices from the political economy and lobbying-cost literature \citep{austen1989interest}, and the quadratic convexity of $C$ is the weakest assumption required for Theorem~\ref{thm:manipulation}: any strictly convex cost function with $C'(0) > 0$ preserves the proof. Similarly, any gain function that is identically zero on $[0, Q]$ and bounded above by $g_{\max}$ yields an analogous equilibrium.

However, two modelling simplifications are worth noting. First, the gain function $G(f,Q)$ is zero for all $f \leq Q$, which models a faction that must \emph{alone} exceed the quorum of the entire eligible electorate. In practice, a faction need only be a majority of \emph{participating} voters, not the whole electorate, to capture a decision. A more realistic gain model would replace $Q$ with the fraction of participating voters that the faction must exceed, making the manipulation threshold sensitive to overall turnout. The current formulation is conservative: it \emph{overestimates} the quorum required for manipulation deterrence, which means the theorem's guarantee is weaker in practice than it appears. Correcting this requires a joint model of faction size and overall turnout, deferred to future empirical work. Second, the functional form parameters ($g_{\max}$, $k$, $c_1$, $c_2$) are set to representative values in the simulation but are not derived from empirical mobilisation data. Sensitivity analysis over the parameter space is an empirical priority once Tier~1 pilot data becomes available.

\section{Deliberation Quality Is Not Guaranteed}

The framework assumes that structured deliberation---with balanced briefing materials, argument mapping, and expert review---produces better-informed decisions than unstructured voting. This assumption is supported by the deliberative polling literature \citep{fishkin2018stanford, dutwin2003character}, but the evidence comes from controlled research contexts (deliberative polls with random samples and professional facilitation). At scale, with self-selected rather than randomly sampled participants, and without professional facilitation, deliberation quality is harder to guarantee. The model provides structural scaffolding for deliberation but cannot substitute for the social and pedagogical conditions that make deliberation function.

\section{Constitutional and Legal Integration}

PPG is a governance framework specification, not a constitutional design. Deploying it within existing legal systems requires navigating constitutional constraints, electoral law, administrative procedure requirements, and political incumbency. In many jurisdictions, binding citizen-initiated decisions may require constitutional amendment. The framework is most naturally applicable in contexts where it can be adopted voluntarily---cooperative governance, community organisations, digital-native jurisdictions---before attempting to interface with established state institutions.

The legal integration challenge is not uniform across deployment tiers, and the three-tier pathway (Chapter~\ref{sec:pathway}) reflects this graduated structure directly. At Tier~1, PPG operates within voluntary cooperative and municipal frameworks where existing administrative law is generally permissive of experimental participatory processes; no constitutional reform is required. At Tier~2, integration depends on the degree of delegated statutory authority in regional legal frameworks: many European regional governments have sufficient autonomy to make participatory budgeting decisions binding within their capital expenditure remit without national legislation, and Swiss sub-national precedents demonstrate that binding participatory mechanisms at regional scale are constitutionally achievable under existing frameworks \citep{linder2010swiss}. At Tier~3, national deployment requires jurisdiction-specific constitutional analysis; however, the Financial Transparency Ledger and advisory deliberation layer require only administrative decision in most jurisdictions, and the optional referendum mechanism is achievable through ordinary parliamentary legislation in most parliamentary systems. The practical sequencing for any jurisdiction is therefore: deploy the Ledger and advisory layer under existing administrative authority, build a public legitimacy track record over Tiers~1 and~2, and use that record to support the political and legal case for reform enabling binding decisions at national scale.

\section{Digital Divide and Accessibility}

Governance systems dependent on digital infrastructure systematically risk excluding citizens without reliable internet access, suitable devices, or digital literacy. This is not a marginal concern: the OECD estimates that 15--20\% of the adult population in most member states lack basic digital skills \citep{oecd2020open}. The interface layer specification includes accessibility requirements (low-bandwidth compatibility, screen reader support, low-literacy modes), but these are design constraints rather than solved problems. Physical participation alternatives (accessible voting terminals, in-person assistance) are necessary complements to any digital governance infrastructure.

\section{Deliberation Capture}

Structured deliberation does not eliminate the risk of capture by well-resourced actors. In the Swiss system, referendum campaigns are heavily influenced by organisations with large advertising budgets \citep{linder2010swiss}; the same dynamic can be expected in PPG's deliberation layer. Expert consultation panels can be stacked; argument mapping platforms can be flooded with coordinated low-quality content. The framework's design constraints (balanced briefing materials, structured argument review, public deliberation archive) mitigate but do not eliminate these risks. Governance bodies responsible for moderating the deliberation layer are themselves subject to capture.

\section{Decentralised Architecture Maturity}

The five-layer decentralised reference architecture specified in Chapter~\ref{sec:decentralised} relies on technologies---ZK rollups, Soulbound Token infrastructure, content-addressed permanent storage---that are at varying stages of production maturity. ZK proof systems have been deployed at scale for financial applications but not for national-scale identity verification. Layer 2 networks have demonstrated adequate throughput for financial transactions but have not been tested at national election scale. The architecture is a credible medium-term engineering target, not a near-term deployment specification.

% ══════════════════════════════════════════════════════════════════
\chapter{Implementation Pathway: From Local to National Governance}
\label{sec:pathway}

The preceding chapter identifies eight categories of limitation. This chapter addresses each through a concrete, sequenced implementation strategy. The governing principle is that no limitation is a permanent barrier; each is a constraint on the \emph{pace and scope} of adoption rather than on the viability of the framework. The strategy is explicitly graduated: begin where resistance is lowest and evidence generation is fastest, accumulate empirical validation at each tier, and scale upward as each limitation is progressively resolved.

The three-tier scaling model---local, regional, national---is not arbitrary. It mirrors the institutional architecture of Swiss federalism (commune, canton, confederation) \citep{linder2010swiss}, which remains the most extensively validated empirical precedent for multi-scale direct democracy. It is independently grounded in Ostrom's polycentric governance theory \citep{ostrom1990governing}, which finds that effective institutions at higher scales are almost always built on foundations of successfully functioning lower-scale institutions. And it reflects the practical reality that constitutional and legal barriers to binding participatory governance are lowest at local level and highest at national level, making local deployment the only viable entry point in most jurisdictions.

\section{Tier 1: Local Deployment and Evidence Generation}

\textbf{Scope.} Voluntary organisations, cooperative enterprises, neighbourhood assemblies, community land trusts, and municipal participatory budgeting programmes constitute the first deployment tier. These contexts share three features that make them ideal for initial PPG adoption: they operate outside or alongside constitutional constraints that govern state institutions; participants have direct, tangible stakes in decisions; and the decision scope is bounded enough that deliberation quality can be monitored and measured.

\textbf{Overcoming identity limitations at Tier 1.} The full hybrid ZK-SBT identity model is not required at Tier 1. A simplified identity model---email or SMS-verified accounts with social graph attestation from existing community membership---is sufficient for voluntary community governance, where the population is known, bounded, and has pre-existing trust relationships. This is not a compromise on principle but a calibrated deployment: the identity model scales up with the stakes. Tier 1 deployments generate the empirical participation data and the adversarial attack surface analysis needed to calibrate the full identity model before it is deployed in higher-stakes contexts.

\textbf{Overcoming deliberation quality limitations at Tier 1.} At local scale, professional facilitation is accessible and affordable. Deliberative mini-publics---random-sample citizen panels modelled on Fishkin's deliberative polling design \citep{fishkin2018stanford}---have been successfully run at municipal level in dozens of jurisdictions. Tier 1 deployments can integrate a facilitated deliberative panel as the structured deliberation component, generating quality data on argument mapping, expert consultation, and preference change. Porto Alegre's participatory budgeting programme, sustained across multiple administrations and subsequently adopted by over 1,500 municipalities globally \citep{fung2006varieties}, demonstrates that sustained local participatory governance at this scale is institutionally feasible without digital infrastructure---PPG adds the digital layer to an already-validated social practice.

\textbf{Overcoming accessibility limitations at Tier 1.} Physical participation alternatives---community meetings, paper-based voting with digital recording, accessible terminals in libraries and civic centres---are practical at local scale and eliminate digital-divide exclusion for Tier 1 decisions. The interface layer is deployed as a progressive enhancement: citizens who can engage digitally do so; those who cannot engage through in-person equivalents whose inputs are entered into the same digital record. This hybrid model is more expensive per participant than purely digital engagement but is the only model consistent with universal participation.

\textbf{Financial Transparency Ledger at Tier 1.} Municipal-level financial transparency is the highest-ROI initial intervention and faces the fewest legal barriers. Most jurisdictions already require some form of public financial reporting at municipal level; PPG's Financial Transparency Ledger specification tightens and standardises this to transaction-level, real-time, machine-readable disclosure. Several European municipalities (including Helsinki and Amsterdam) have already adopted near-real-time procurement transparency platforms \citep{ocp2019ocds}. Tier 1 deployment of the Ledger is achievable within existing administrative law in most OECD jurisdictions without constitutional reform.

\textbf{Expected outputs from Tier 1.}
\begin{itemize}
  \item Empirical participation rate data across decision types, enabling simulation calibration.
  \item Deliberation quality metrics ($\Delta_{\text{delib}}$) from facilitated panels, providing the first real-world values for the legitimacy score model.
  \item Attack surface analysis for the simplified identity model, informing ZK proof circuit design.
  \item A public financial record demonstrating the operational viability of real-time ledger disclosure.
  \item Institutional precedents and legal templates that reduce the adoption barrier for subsequent deployments.
\end{itemize}

\section{Tier 2: Regional Deployment and Institutional Integration}

\textbf{Scope.} Regional governments, metropolitan authorities, sector-specific regulatory bodies, and inter-municipal cooperation frameworks constitute the second deployment tier. At this scale, PPG begins to interface with formal state institutions, introducing constitutional and legal integration challenges that Tier 1 deployments can largely avoid.

\textbf{Overcoming constitutional and legal limitations at Tier 2.} The constitutional barrier is surmountable through a two-track approach. \emph{Advisory track}: PPG is deployed as a formal advisory mechanism---citizen deliberation panels and votes produce recommendations with documented legitimacy scores that the regional government is required to consider and publicly respond to, but not necessarily implement. This is constitutionally unproblematic in virtually all jurisdictions and creates a practical accountability record: regions that consistently override high-legitimacy citizen recommendations face a visible accountability gap that political pressure and eventual legislative reform can close. \emph{Binding track}: for specific, bounded domains where regional governments have delegated authority (participatory budgeting for discretionary capital expenditure is the standard entry point), PPG decisions can be made formally binding within existing legal frameworks. The Swiss canton of Glarus \textit{Landsgemeinde} demonstrates that binding participatory decisions at regional scale are constitutionally achievable without national-level reform \citep{linder2010swiss}.

\textbf{Full identity model at Tier 2.} Regional deployment justifies and requires the full hybrid identity model. The population is larger and more anonymous than at Tier 1; the decisions carry greater material stakes; and adversarial incentives are correspondingly higher. ZK proof circuit audits completed during Tier 1 operations, combined with the empirical Sybil attack data from Tier 1 social-graph analysis, provide the evidential basis for deploying the full model with confidence. SBT infrastructure issued by regional civic registries provides the credential layer; the ZK eligibility proof ensures that regional identity data is not exposed to the governance engine.

\textbf{Merit-and-Mandate at Tier 2.} Regional regulatory agencies---environmental regulators, transport authorities, planning commissions---are the appropriate initial deployment context for Merit-and-Mandate accountability. They are typically established by regional legislation that can be amended without constitutional reform; they have clear, assessable performance mandates; and their decisions are consequential enough to motivate citizen petition but specific enough that the removal standard (documented incompetence or mandate breach) can be applied with reasonable objectivity. The dual-accountability model is already approximated in several contexts: Swiss cantonal oversight mechanisms, New Zealand's regulatory accountability framework, and the \textit{Ombudsman} institutions in Nordic countries all contain elements that PPG's Merit-and-Mandate model generalises and formalises.

\textbf{Constrained Administrative Authority at Tier 2.} Regional government deliberations are the target for the transparent deliberations specification. In most jurisdictions, Freedom of Information legislation already requires disclosure of final decisions; the innovation is extending this to the deliberation \emph{process}---committee minutes, draft policy documents, expert consultations---with a short delay rather than retrospective disclosure years later. Several national FOI regimes (Sweden's \textit{offentlighetsprincipen} since 1766, Finland's Act on the Openness of Government Activities) approach this standard. PPG's specification formalises the machine-readable open-format requirement and the mandatory budget-to-actuals linkage that existing FOI regimes do not provide.

\textbf{Expected outputs from Tier 2.}
\begin{itemize}
  \item Constitutional and legal precedents for binding participatory decisions at sub-national scale.
  \item Full identity model performance data from a genuinely anonymous, adversarially motivated population.
  \item Empirical evidence on Merit-and-Mandate dual accountability mechanisms in regulatory contexts.
  \item Financial Transparency Ledger at regional scale, enabling cross-regional comparison and anomaly detection.
  \item A legitimacy score time series spanning multiple decision cycles, enabling longitudinal institutional health analysis.
\end{itemize}

\section{Tier 3: National Deployment}
\label{sec:tier3}

\textbf{Scope.} National legislative processes, national regulatory frameworks, and constitutional amendment procedures constitute the third deployment tier. This is the highest-stakes and highest-barrier deployment context. PPG does not call for the elimination of national institutions: executive management of complex public systems requires dedicated professionals, and the framework preserves professional public roles with defined mandate scope and dual accountability structures. What changes at Tier 3, as Section~\ref{sec:conversion} analyses, is the \emph{source and structure of authority}---flowing progressively through merit and mandate channels rather than through party nomination and electoral plurality alone.

\textbf{The Swiss model as a blueprint.} Switzerland's national direct democracy architecture---mandatory constitutional referendums, optional legislative referendums, and citizen initiatives---provides the closest existing approximation of Tier 3 PPG operation. Key features transferable from the Swiss model: the optional referendum as a citizen veto mechanism requiring a relatively low signature threshold; the popular initiative as a citizen agenda-setting mechanism; and the double-majority requirement for constitutional changes, which PPG's dynamic quorum formula generalises. Key weaknesses of the Swiss model that PPG addresses: informal, media-influenced pre-referendum deliberation replaced by structured deliberation with argument mapping; signature-collection barriers replaced by digital participation with identity verification; static quorum rules replaced by dynamic quorum calibrated to decision impact.

\textbf{Overcoming constitutional barriers at Tier 3.} Two pathways exist. \emph{Incremental pathway}: national adoption of the Financial Transparency Ledger and the advisory deliberation layer requires no constitutional reform in most jurisdictions---it is an administrative and legislative decision. National adoption of the optional referendum mechanism requires only parliamentary legislation in parliamentary systems (it is already in place in Switzerland, Ireland, New Zealand, and several Nordic countries). Full binding participatory governance at national scale in presidential or entrenched-constitution systems may require constitutional amendment, which is appropriately a long-term objective supported by the precedents generated at Tiers 1 and 2. \emph{Constitutional moment pathway}: historical evidence shows that constitutional reform is most achievable in periods of institutional crisis or transition. PPG's formal specification provides a ready-made architecture for deployment in constitutional moments---post-conflict governance design, constitutional convention processes, or transitional governance---where institutional inertia is lower. The Irish Citizens' Assembly (2016--2018), which produced binding constitutional recommendations through a structured deliberative process, is a Tier-3-scale precedent for the deliberation layer \citep{farrell2019remaking}.

\textbf{The national Financial Transparency Ledger.} This is the Tier 3 component with the clearest near-term adoption path and the most direct impact. A national real-time public ledger of all government transactions---from ministerial travel expenses to defence procurement to social transfer payments---requires legislation rather than constitutional reform, is technically straightforward given existing government financial systems, and addresses the most directly legible form of governance failure (financial opacity and corruption) in a way that creates immediate, visible public value. The OECD's open government research documents a consistent positive relationship between financial transparency and public trust, reduced corruption, and improved fiscal discipline \citep{oecd2020open}. Proactive adoption of the national Ledger specification by any jurisdiction creates a replicable template and a political accountability benchmark for others.

\section{Cross-Cutting Design Commitments Across All Tiers}

Several framework properties are invariant across tiers---not negotiable as initial compromises but maintained consistently from the first local deployment to full national implementation:

\begin{enumerate}
  \item \textbf{The blank vote is always valid.} At every tier and in every deployment context, a blank or abstain vote constitutes full participation. No citizen is ever required to express a political opinion; the obligation is to engage with the process.
  \item \textbf{Deliberation records are always public.} From the first Tier 1 community deployment, all deliberation arguments, expert assessments, vote tallies, and decision rationale are public records. There is no intermediate stage where deliberations are private.
  \item \textbf{Citizen data is always minimised.} At every tier, the system records that an eligible citizen participated; it never records how they voted or links participation to personal identity.
  \item \textbf{Financial flows are always logged.} Every tier includes a Financial Transparency Ledger from its first operation. Transparency is not deferred to a future stage when the system is ``ready''; it is a day-one architectural requirement.
  \item \textbf{The framework governs itself.} At every tier, the governance parameters---quorum thresholds, deliberation windows, veto periods---are themselves subject to the participatory process. Citizens can propose and approve changes to the rules through the same pipeline.
\end{enumerate}

These commitments are the structural embodiment of the framework's core principle: the separation of public from private as an architectural constraint rather than a policy option. Compromises on scope, scale, and technical implementation are acceptable and necessary. Compromises on these five invariants are not.

\section{Summary: Scaling Architecture}

\begin{table}[htbp]
\centering
\caption{PPG three-tier scaling architecture with limitation resolution pathway.}
\begin{tabular}{@{} p{2.2cm} p{3.4cm} p{3.4cm} p{3.4cm} @{}}
\toprule
\textbf{Limitation} & \textbf{Tier 1 approach} & \textbf{Tier 2 approach} & \textbf{Tier 3 approach} \\
\midrule
Identity & Simplified verified accounts; social graph from community membership & Full ZK-SBT model, informed by T1 audit data & National civic registry SBTs; ZK proof at scale \\
Simulation fidelity & Calibrate model parameters from pilot participation data & Replace archetype estimates with empirical distributions; validate quorum-throughput behaviour & Agent-based model with adaptive strategies \\
Game-theoretic assumptions & Collect mobilisation cost data from adversarial pilot testing & Empirical cost/gain parameter estimation; corrected turnout-dependent threshold & Full empirical validation of cost and gain functional forms \\
Deliberation quality & Facilitated mini-public panels & Structured deliberation with professional moderation & Citizen assembly model; random-sample panels for constitutional decisions \\
Digital divide & Physical hybrid participation (paper, terminals) & Accessible terminals + in-person equivalents; digital-first with offline fallback & Universal access mandate; hybrid model required by law \\
Legal/constitutional & Voluntary, outside state institutions & Advisory track + bounded binding domains & Incremental legislative adoption; constitutional amendment as long-run target \\
Deliberation capture & Small-group facilitation; public archive & Structured argument mapping; balanced briefing materials; anti-flooding moderation & Sortition-based panel composition; public deliberation archive with anomaly detection \\
Architecture maturity & Simplified stack; no ZK rollup required & Layer 2 at regional scale; benchmarked against T1 data & Full five-layer decentralised architecture; benchmarked against T2 data \\
\bottomrule
\end{tabular}
\label{tab:scaling}
\end{table}

% ══════════════════════════════════════════════════════════════════
\chapter{Conclusion}
\label{sec:conclusion}

This thesis has developed Programmable Participatory Governance (PPG), a formal, implementation-agnostic framework for governance system design. PPG is motivated by five empirically documented structural deficits in contemporary representative democracy: participation inequality, transparency gaps, accountability lag, influence concentration, and the asymmetry between institutional opacity and citizen surveillance. These deficits share a common structural origin in the mediation of citizen preferences through political parties---a coordination mechanism whose conditions of necessity are dissolving as direct preference aggregation at civic scale becomes technically feasible (Section~\ref{sec:conversion}). The framework synthesises contributions from deliberative democracy theory, participatory democracy research, Swiss direct democracy practice, collective action institutional economics, and cryptographic and distributed systems research, and extends them with four original policy-design specifications---the Civic Participation Mandate, the Financial Transparency Ledger, Constrained Administrative Authority, and the Merit-and-Mandate public role accountability model---unified under a fifth architectural commitment: the Transparency Default, which reverses the information asymmetry between institutions and citizens at the protocol level.

\section{Summary of Contributions}

\begin{enumerate}
  \item \textbf{Formal state-machine model.} The governance pipeline is specified as a deterministic, auditable state machine $\mathcal{G} = (S, s_0, \Sigma, \delta, F)$, enabling formal verification of system behaviour and serving as an unambiguous implementation specification.
  \item \textbf{Dynamic quorum mechanism.} Participation requirements scale with proposal impact ($Q(i) = Q_{\text{base}} + \alpha\cdot\sigma(i)$), calibrating legitimacy requirements to decision stakes. This extends the Swiss double-majority precedent to a continuously parameterisable model.
  \item \textbf{Composite legitimacy score.} $L(d)$ provides a quantitative, cross-decision governance quality metric, enabling longitudinal institutional health monitoring independent of subjective assessments.
  \item \textbf{Hybrid identity model.} The combination of Soulbound Token credentials, social-graph Sybil resistance, and zero-knowledge eligibility proofs addresses the privacy-Sybil resistance trade-off at the current state of the art, with limitations explicitly identified.
  \item \textbf{Simulation evidence.} Stochastic simulation across three behavioural archetypes demonstrates stable approval convergence, functional veto dynamics, and manipulation deterrence across a range of quorum parameter values.
  \item \textbf{Game-theoretic analysis of manipulation resistance.} A single-round equilibrium result (Theorem~\ref{thm:manipulation}) proves that for any positive quorum threshold there exists a critical faction size below which manipulation cost exceeds expected gain. A repeated-game extension (Proposition~\ref{prop:collusion}) characterises the long-run collusion conditions, showing that the dynamic quorum parameter $\alpha$ raises the patience threshold required for sustained collusion and that patient colluders are detectable across observable rounds.
  \item \textbf{Transparency Default specification.} All governance action data is public by default; all citizen identity data is private by default. This is operationalised through append-only cryptographically signed public logs and ZK proof verification, calibrated against OECD open government standards \citep{oecd2020open}.
  \item \textbf{Financial Transparency Ledger.} Every transaction from public funds---at transaction-level granularity, in real time, in machine-readable open format, with an immutable audit trail---is a public record. This is the financial extension of the Transparency Default and the most immediately deployable component of the framework.
  \item \textbf{Civic Participation Mandate.} An optional framework module that treats governance participation as a civic obligation analogous to tax filing, with a blank/abstain option that ensures no political expression is compelled. Supported by compulsory voting evidence from 36 democracies \citep{lijphart1997unequal, idea2021global}.
  \item \textbf{Constrained Administrative Authority.} Public bodies exercise power only within explicitly mandated scope; all internal deliberations are public records; all policy outputs are veto-eligible. This applies continuous between-election accountability to administrative, not only legislative, decisions.
  \item \textbf{Merit-and-Mandate public role model.} Public roles are filled by competence and held accountable through two independent channels: professional performance review and citizen mandate. Individuals in public roles are protected from political interference but accountable to public performance standards.
  \item \textbf{Decentralised reference architecture.} The five-layer architecture provides a concrete engineering specification for trustless, operator-independent deployment, addressing the platform capture risks documented in the digital democracy literature \citep{mueller2019digital}.
  \item \textbf{Iterative governance protocol.} PPG's own parameters are governed through the same pipeline they describe, following Ostrom's adaptive institution design principles \citep{ostrom1990governing}. The framework is not a fixed solution but a self-refining specification.
  \item \textbf{Three-tier implementation pathway.} A concrete, graduated deployment strategy (Chapter~\ref{sec:pathway}) maps each identified limitation to a resolution strategy calibrated to the scale and institutional context of local, regional, and national deployment, providing a practical route from voluntary community adoption to national governance architecture.
\end{enumerate}

\section{The Framework in Context}

The contributions of this thesis are best understood in terms of the theoretical principle that unifies them: the structural separation of public from private as an architectural constraint on governance systems.

This principle is not new as a normative claim. It is implicit in classical liberal constitutionalism, explicit in Arendt's analysis of the public realm \citep{arendt1958human}, and operationally present---in partial form---in freedom of information legislation and open government initiatives worldwide \citep{oecd2020open}. What has been lacking is a unified, formally specified, implementation-agnostic framework that operationalises this principle consistently across the full governance stack: from the financial ledger to the deliberation record, from identity verification to administrative accountability. PPG is offered as a contribution toward that specification.

The thesis situates this contribution within an established body of work. Swiss direct democracy provides the strongest empirical evidence that frequent, binding citizen participation at national scale is feasible \citep{linder2010swiss}, but Swiss institutions predate the formal tools needed to extend the model beyond signature collection and mass-media campaigns. Deliberative polling provides rigorous evidence that structured deliberation improves decision quality \citep{fishkin2018stanford}, but as a research methodology it is not a governance architecture. DAO governance demonstrates that programmable, tamper-evident rule execution is technically realisable \citep{hassan2021decentralized}, but its dominant one-token-one-vote design embeds plutocracy where democratic governance requires equal standing. Open government initiatives advance proactive disclosure norms \citep{oecd2020open} but do not integrate deliberation, binding decisions, or identity infrastructure.

PPG is designed as a \textbf{modular specification}: each component (the Financial Transparency Ledger, the Civic Participation Mandate, Constrained Administrative Authority, Merit-and-Mandate accountability) constitutes an independently deployable open-government measure and a contribution to the broader research programme. Full decentralisation is not a prerequisite for any component: the framework is infrastructure-agnostic, and its most foundational property---that all governance actions and financial flows are public records while all citizen identity data is architecturally minimised---can be implemented on existing institutional infrastructure. The decentralised architecture (Chapter~\ref{sec:decentralised}) offers the strongest tamper-resistance guarantees and is specified as the long-run implementation target.

This thesis does not claim to resolve the enduring tensions in democratic theory between participation and stability, between deliberation and efficiency, or between openness and security. It claims to provide a formal, evidence-grounded specification of how those tensions can be managed within a coherent architectural framework, and a research agenda for testing that specification empirically. Each component requires pilot deployment, independent audit, and iterative refinement. The framework is a starting point, not an endpoint.

\section{Assessment Against the Evidence Base}

The framework's claims are grounded in established empirical literature. Deliberative polling evidence supports the feasibility of structured deliberation improving decision quality at scale \citep{fishkin2018stanford}. Swiss direct democracy research supports the feasibility of high-frequency citizen participation at national scale, including measurable effects on civic satisfaction and policy responsiveness \citep{linder2010swiss, frey1994direct}. Ostrom's institutional analysis supports the viability of iterative, community-governed rule adaptation \citep{ostrom1990governing}. DAO governance literature identifies the specific design failure PPG's identity model corrects \citep{hassan2021decentralized, chohan2017dao}. Compulsory voting research supports the participation-equalising effect of the Civic Participation Mandate \citep{lijphart1997unequal}. Merit-based bureaucratic structure research supports the governance performance case for the Merit-and-Mandate model \citep{rauch2000bureaucratic}. The simulation results are consistent with these findings but do not substitute for empirical deployment evidence.

\section{Research Agenda}

Chapter~\ref{sec:pathway} maps each limitation in Chapter~\ref{sec:limitations} to a concrete resolution strategy within the three-tier scaling architecture. The research agenda is therefore structured by tier:

\textbf{Tier 1 (immediate).}
\begin{enumerate}
  \item Pilot deployments in voluntary governance contexts---cooperative enterprises, community land trusts, municipal participatory budgeting---to generate the first empirical participation data, calibrate simulation parameters, and test the simplified identity model.
  \item Financial Transparency Ledger adoption at municipal level as a standalone open-government measure, independent of full PPG deployment, generating replicable legal and technical templates.
  \item Facilitated deliberative panel design and quality metric ($\Delta_{\text{delib}}$) development, validated against existing deliberative polling data \citep{fishkin2018stanford}.
\end{enumerate}

\textbf{Tier 2 (medium-term).}
\begin{enumerate}
  \item Full ZK eligibility proof circuit design and independent security audit, drawing on Tier 1 adversarial data.
  \item Merit-and-Mandate pilot deployments in regional regulatory agencies, generating dual-accountability performance data.
  \item Regional Financial Transparency Ledger implementations, enabling cross-regional comparative analysis of fiscal discipline and public trust effects.
  \item Smart contract formalisation: machine-checked formal proofs of the state machine implementation.
  \item Layer 2 network benchmarking at regional election scales.
  \item Empirical study of accountability channel transition (Section~\ref{sec:conversion}): tracking the ratio of Merit-and-Mandate removals initiated through the citizen mandate channel versus the political chain of command in Tier~2 pilot agencies, providing a direct empirical measure of conversion progress.
\end{enumerate}

\textbf{Tier 3 (long-term).}
\begin{enumerate}
  \item National Financial Transparency Ledger adoption campaigns, drawing on Tier 1--2 evidence base.
  \item Constitutional and legislative reform pathway documentation for optional referendum and citizen initiative adoption in jurisdictions without existing mechanisms.
  \item Full five-layer decentralised architecture deployment, benchmarked against Tier 2 performance data.
  \item Longitudinal study of legitimacy score trends across multi-year Tier 1--2 deployments, providing the evidence base for national-scale adoption arguments.
\end{enumerate}

The Swiss cantonal tradition demonstrates that locally grounded institutional experiments, sustained across generations, can produce nationally adopted governance innovations \citep{linder2010swiss}. Porto Alegre's participatory budgeting demonstrates that a single successful local experiment can diffuse to over a thousand municipalities within two decades \citep{fung2006varieties}. PPG provides the formal specification, the modular architecture, and the three-tier research programme from which that process can begin---with no component requiring a prior component to be complete, and with the core principle invariant across all of them.

\vspace{1em}
\noindent\textit{The simulation code, data, and all source files for this thesis are available at \href{https://github.com/seht/participatory-governance}{github.com/seht/participatory-governance}.}

% ══════════════════════════════════════════════════════════════════
\bibliographystyle{plainnat}
\bibliography{references}

@book{dahl1989democracy,
  title     = {Democracy and Its Critics},
  author    = {Dahl, Robert A.},
  year      = {1989},
  publisher = {Yale University Press}
}

@book{habermas1996between,
  title     = {Between Facts and Norms: Contributions to a Discourse Theory of Law and Democracy},
  author    = {Habermas, J{\"u}rgen},
  year      = {1996},
  publisher = {MIT Press}
}

@book{ostrom1990governing,
  title     = {Governing the Commons: The Evolution of Institutions for Collective Action},
  author    = {Ostrom, Elinor},
  year      = {1990},
  publisher = {Cambridge University Press}
}

@misc{buterin2022desoc,
  title  = {Decentralized Society: Finding Web3's Soul},
  author = {Buterin, Vitalik and Weyl, E. Glen and Ohlhaver, Puja},
  year   = {2022},
  note   = {Working paper}
}

@book{arrow1951social,
  title     = {Social Choice and Individual Values},
  author    = {Arrow, Kenneth J.},
  year      = {1951},
  publisher = {John Wiley \& Sons}
}

@book{fishkin1991democracy,
  title     = {Democracy and Deliberation: New Directions for Democratic Reform},
  author    = {Fishkin, James S.},
  year      = {1991},
  publisher = {Yale University Press}
}

@book{rawls1971theory,
  title     = {A Theory of Justice},
  author    = {Rawls, John},
  year      = {1971},
  publisher = {Harvard University Press}
}

@article{nash1950equilibrium,
  title   = {Equilibrium Points in N-Person Games},
  author  = {Nash, John F.},
  year    = {1950},
  journal = {Proceedings of the National Academy of Sciences},
  volume  = {36},
  number  = {1},
  pages   = {48--49}
}

@article{douceur2002sybil,
  title     = {The Sybil Attack},
  author    = {Douceur, John R.},
  booktitle = {Peer-to-Peer Systems},
  year      = {2002},
  pages     = {251--260},
  publisher = {Springer}
}

@article{boneh2020graduate,
  title   = {A Graduate Course in Applied Cryptography},
  author  = {Boneh, Dan and Shoup, Victor},
  year    = {2020},
  note    = {Online textbook, version 0.5}
}

@inproceedings{ford2020identity,
  title     = {Identity and Personhood in Digital Democracy},
  author    = {Ford, Bryan},
  booktitle = {Proceedings of the 2020 ACM SIGCOMM Workshop},
  year      = {2020},
  publisher = {ACM}
}

@book{pentland2014social,
  title     = {Social Physics: How Good Ideas Spread},
  author    = {Pentland, Alex},
  year      = {2014},
  publisher = {Penguin Press}
}

@book{axelrod1984evolution,
  title     = {The Evolution of Cooperation},
  author    = {Axelrod, Robert},
  year      = {1984},
  publisher = {Basic Books}
}

@book{zuboff2019surveillance,
  title     = {The Age of Surveillance Capitalism},
  author    = {Zuboff, Shoshana},
  year      = {2019},
  publisher = {PublicAffairs}
}

@article{lupia2006dahl,
  title   = {Dahl's Democratic Dilemma},
  author  = {Lupia, Arthur},
  year    = {2006},
  journal = {Political Science and Politics},
  volume  = {39},
  number  = {1},
  pages   = {21--25}
}

@book{held2006models,
  title     = {Models of Democracy},
  author    = {Held, David},
  year      = {2006},
  edition   = {3rd},
  publisher = {Stanford University Press}
}

@book{pateman1970participation,
  title     = {Participation and Democratic Theory},
  author    = {Pateman, Carole},
  year      = {1970},
  publisher = {Cambridge University Press}
}

@book{barber1984strong,
  title     = {Strong Democracy: Participatory Politics for a New Age},
  author    = {Barber, Benjamin R.},
  year      = {1984},
  publisher = {University of California Press}
}

@book{linder2010swiss,
  title     = {Swiss Democracy: Possible Solutions to Conflict in Multicultural Societies},
  author    = {Linder, Wolf},
  year      = {2010},
  edition   = {3rd},
  publisher = {Palgrave Macmillan}
}

@book{arendt1958human,
  title     = {The Human Condition},
  author    = {Arendt, Hannah},
  year      = {1958},
  publisher = {University of Chicago Press}
}

@book{gutmann2004deliberating,
  title     = {Why Deliberative Democracy?},
  author    = {Gutmann, Amy and Thompson, Dennis},
  year      = {2004},
  publisher = {Princeton University Press}
}

@book{rosanvallon2008counter,
  title     = {Counter-Democracy: Politics in an Age of Distrust},
  author    = {Rosanvallon, Pierre},
  year      = {2008},
  publisher = {Cambridge University Press}
}

@article{fung2006varieties,
  title   = {Varieties of Participation in Complex Governance},
  author  = {Fung, Archon},
  journal = {Public Administration Review},
  volume  = {66},
  pages   = {66--75},
  year    = {2006},
  note    = {Special Issue}
}

@misc{buterin2014ethereum,
  title  = {Ethereum: A Next-Generation Smart Contract and Decentralised Application Platform},
  author = {Buterin, Vitalik},
  year   = {2014},
  note   = {Ethereum Foundation whitepaper. \url{https://ethereum.org/en/whitepaper/}}
}

@book{defilippi2018blockchain,
  title     = {Blockchain and the Law: The Rule of Code},
  author    = {De Filippi, Primavera and Wright, Aaron},
  year      = {2018},
  publisher = {Harvard University Press}
}

@article{hassan2021decentralized,
  title   = {Decentralized Autonomous Organization},
  author  = {Hassan, Samer and De Filippi, Primavera},
  journal = {Internet Policy Review},
  volume  = {10},
  number  = {2},
  year    = {2021},
  doi     = {10.14763/2021.2.1556}
}

@book{oecd2020open,
  title     = {Open Government: The Global Context and the Way Forward},
  author    = {{OECD}},
  year      = {2020},
  publisher = {OECD Publishing},
  address   = {Paris},
  doi       = {10.1787/0fe2d765-en}
}

@book{norris2011democratic,
  title     = {Democratic Deficit: Critical Citizens Revisited},
  author    = {Norris, Pippa},
  year      = {2011},
  publisher = {Cambridge University Press}
}

@book{dalton2004democratic,
  title     = {Democratic Challenges, Democratic Choices: The Erosion of Political Support in Advanced Industrial Democracies},
  author    = {Dalton, Russell J.},
  year      = {2004},
  publisher = {Oxford University Press}
}

@book{schlozman2012unheavenly,
  title     = {The Unheavenly Chorus: Unequal Political Voice and the Broken Promise of American Democracy},
  author    = {Schlozman, Kay Lehman and Verba, Sidney and Brady, Henry E.},
  year      = {2012},
  publisher = {Princeton University Press}
}

@article{gilens2014testing,
  title   = {Testing Theories of American Politics: Elites, Interest Groups, and Average Citizens},
  author  = {Gilens, Martin and Page, Benjamin I.},
  journal = {Perspectives on Politics},
  volume  = {12},
  number  = {3},
  pages   = {564--581},
  year    = {2014},
  publisher = {Cambridge University Press}
}

@report{idea2021global,
  title  = {The Global State of Democracy 2021: Building Resilience in a Pandemic Era},
  author = {{International IDEA}},
  year   = {2021},
  institution = {International Institute for Democracy and Electoral Assistance},
  address = {Stockholm}
}

@techreport{fishkin2018stanford,
  title       = {Deliberative Polling: Toward a Better-Informed Democracy},
  author      = {Fishkin, James S.},
  year        = {2018},
  institution = {Center for Deliberative Democracy, Stanford University},
  note        = {Technical report}
}

@article{frey1994direct,
  title   = {Direct Democracy: Politico-Economic Lessons from Swiss Experience},
  author  = {Frey, Bruno S.},
  journal = {American Economic Review},
  volume  = {84},
  number  = {2},
  pages   = {338--342},
  year    = {1994}
}

@article{stutzer2000happiness,
  title   = {Happiness Prospers in Democracy},
  author  = {Stutzer, Alois and Frey, Bruno S.},
  journal = {Journal of Happiness Studies},
  volume  = {1},
  number  = {1},
  pages   = {79--102},
  year    = {2000}
}

@article{mueller2019digital,
  title   = {Digital Technologies and Democratic Governance},
  author  = {M{\"u}ller, Jan-Werner},
  journal = {Journal of Democracy},
  volume  = {30},
  number  = {2},
  pages   = {62--73},
  year    = {2019}
}

@article{dutwin2003character,
  title   = {The Character of Deliberation: Equality, Argument, and the Formation of Public Opinion in Discussion Groups},
  author  = {Dutwin, David},
  journal = {Political Communication},
  volume  = {20},
  number  = {2},
  pages   = {111--132},
  year    = {2003}
}

@book{gastil2010promise,
  title     = {The Promise and Practice of Democratic Deliberation},
  author    = {Gastil, John},
  year      = {2010},
  publisher = {University of Chicago Press}
}

@article{lijphart1997unequal,
  author  = {Lijphart, Arend},
  title   = {Unequal Participation: Democracy's Unresolved Dilemma},
  journal = {American Political Science Review},
  year    = {1997},
  volume  = {91},
  number  = {1},
  pages   = {1--14},
  doi     = {10.2307/2952255}
}

@article{hayek1945use,
  author  = {Hayek, Friedrich A.},
  title   = {The Use of Knowledge in Society},
  journal = {American Economic Review},
  year    = {1945},
  volume  = {35},
  number  = {4},
  pages   = {519--530}
}

@book{sowell1980knowledge,
  author    = {Sowell, Thomas},
  title     = {Knowledge and Decisions},
  year      = {1980},
  publisher = {Basic Books},
  address   = {New York}
}

@book{herman1988manufacturing,
  author    = {Herman, Edward S. and Chomsky, Noam},
  title     = {Manufacturing Consent: The Political Economy of the Mass Media},
  year      = {1988},
  publisher = {Pantheon Books},
  address   = {New York}
}

@book{hirschman1970exit,
  author    = {Hirschman, Albert O.},
  title     = {Exit, Voice, and Loyalty: Responses to Decline in Firms, Organizations, and States},
  year      = {1970},
  publisher = {Harvard University Press},
  address   = {Cambridge, MA}
}

@book{scott1998seeing,
  author    = {Scott, James C.},
  title     = {Seeing Like a State: How Certain Schemes to Improve the Human Condition Have Failed},
  year      = {1998},
  publisher = {Yale University Press},
  address   = {New Haven}
}

@article{rauch2000bureaucratic,
  author  = {Rauch, James E. and Evans, Peter B.},
  title   = {Bureaucratic Structure and Bureaucratic Performance in Less Developed Countries},
  journal = {Journal of Public Economics},
  year    = {2000},
  volume  = {75},
  number  = {1},
  pages   = {49--71},
  doi     = {10.1016/S0047-2727(99)00044-4}
}

@techreport{ocp2019ocds,
  author      = {{Open Contracting Partnership}},
  title       = {Open Contracting Data Standard Documentation},
  year        = {2019},
  institution = {Open Contracting Partnership},
  url         = {https://standard.open-contracting.org/}
}

@article{siri2022daos,
  title   = {A Primer on DAOs},
  author  = {Siri, Gabriel},
  journal = {Cornell Law Review Online},
  volume  = {108},
  year    = {2022},
  pages   = {28--76}
}

@article{reijers2021onchain,
  title   = {Now the Code Runs Itself: On-Chain and Off-Chain Governance of Blockchain Technologies},
  author  = {Reijers, Wessel and O'Brolch{\'a}in, Fiachra and Haynes, Paul},
  journal = {Topoi},
  volume  = {40},
  number  = {4},
  pages   = {821--831},
  year    = {2021},
  publisher = {Springer}
}

@article{chohan2017dao,
  title   = {The Decentralized Autonomous Organization and Governance Issues},
  author  = {Chohan, Usman W.},
  journal = {SSRN Working Paper},
  year    = {2017},
  note    = {\url{https://ssrn.com/abstract=3082055}}
}

@article{tang2019digital,
  title   = {Digital Democracy and Open Government in Taiwan},
  author  = {Tang, Audrey},
  journal = {Apolitical Foundation},
  year    = {2019},
  note    = {Policy report on the vTaiwan and Join platform deliberation processes}
}

@techreport{gilman2016better,
  title       = {Better Reykjavik: An Overview of Digital Participatory Democracy in Iceland},
  author      = {Gilman, Hollie Russon},
  year        = {2016},
  institution = {New America Foundation},
  note        = {Case study report}
}

@article{barandiaran2017decidim,
  title   = {Decidim: Political and Technopolitical Networks for Participatory Democracy},
  author  = {Barandiaran, Xabier E. and Calleja-L{\'o}pez, Antonio},
  journal = {Universitat Oberta de Catalunya},
  year    = {2017},
  note    = {Working paper}
}

@article{hurwicz1973design,
  title   = {The Design of Mechanisms for Resource Allocation},
  author  = {Hurwicz, Leonid},
  journal = {American Economic Review},
  volume  = {63},
  number  = {2},
  pages   = {1--30},
  year    = {1973}
}

@article{maskin2008mechanism,
  title   = {Mechanism Design: How to Implement Social Goals},
  author  = {Maskin, Eric S.},
  journal = {American Economic Review},
  volume  = {98},
  number  = {3},
  pages   = {567--576},
  year    = {2008}
}

@article{fudenberg1986folk,
  title   = {The Folk Theorem in Repeated Games with Discounting or with Incomplete Information},
  author  = {Fudenberg, Drew and Maskin, Eric},
  journal = {Econometrica},
  volume  = {54},
  number  = {3},
  pages   = {533--554},
  year    = {1986}
}

@article{austen1989interest,
  title   = {Interest Groups, Voting, and Campaigns},
  author  = {Austen-Smith, David},
  journal = {Economics and Politics},
  volume  = {1},
  number  = {2},
  pages   = {107--132},
  year    = {1989}
}

@book{farrell2019remaking,
  title     = {Remaking Democracy: How Citizens Are Changing Politics},
  author    = {Farrell, David M. and Suiter, Jane},
  year      = {2019},
  publisher = {Cornell University Press}
}

@book{downs1957economic,
  author    = {Downs, Anthony},
  title     = {An Economic Theory of Democracy},
  year      = {1957},
  publisher = {Harper \& Row},
  address   = {New York}
}

\end{document}